\documentclass[11pt]{article}

\usepackage[margin=1in]{geometry}
\usepackage{amsmath,amssymb,amsthm,mathtools}
\usepackage{enumitem}
\usepackage{xcolor}
\usepackage{hyperref}
\usepackage{color}
\usepackage{biblatex}
\usepackage[normalem]{ulem} 
 
\hypersetup{colorlinks=true,linkcolor=blue!60!black,citecolor=blue!60!black,urlcolor=blue!60!black}

\newtheorem{theorem}{Theorem}[section]
\newtheorem{proposition}[theorem]{Proposition}
\newtheorem{lemma}[theorem]{Lemma}
\newtheorem{corollary}[theorem]{Corollary}
\newtheorem{assumption}[theorem]{Assumption}
\newtheorem{remark}[theorem]{Remark}
\newtheorem{definition}[theorem]{Definition}

\newcommand{\X}{X}
\newcommand{\A}{A}
\newcommand{\Pcal}{\mathcal P}
\newcommand{\R}{\mathbb R}
\newcommand{\dd}{\,\mathrm d}
\newcommand{\norm}[1]{\left\lVert #1 \right\rVert}

\newcommand{\spn}{\operatorname{span}}
\newcommand{\esssup}{\operatorname*{ess\,sup}}

\bibliography{references}

\author{Junji Yan$^{\dagger}$ \and U\u{g}ur Ayd{\i}n$^{\dagger}$ \and Tamer Ba\c{s}ar\thanks{Authors are with the Department of Electrical and Computer Engineering, and the Coordinated Science Laboratory, University of Illinois Urbana-Champaign, Urbana, IL, 61801
        {\tt\small \{junjiy2,uaydin2,basar1\}@illinois.edu}
         The first two authors contributed equally to this work. }%
}
\title{Horizon-Independent Contraction for Continuous-Time Discounted Regularized Mean-Field Games\thanks{Research of the first authors was supported in part by the IBM-ILLINOIS Discovery Accelerator Institute (IIDAI) under Grant \#122955. Research of the second and third authors was supported in part by the AFOSR Grant FA9550-24-1-0152.}}
\date{}
\begin{document}
\maketitle

\begin{abstract}
    We study contraction properties of non-stationary continuous-time mean-field games (MFGs) under discounting and entropy regularization. The state of the representative agent evolves according to a controlled continuous-time Markov chain, and both the state and action spaces are finite. 
    In contrast to the undiscounted case, we show that, under a sufficiently large discount rate, finite-horizon MFGs admit a horizon-independent contraction condition, which also coincides with the corresponding infinite-horizon non-stationary contraction condition. As a byproduct, we obtain an explicit convergence rate between finite- and infinite-horizon mean-field equilibria.  For each finite horizon, we further derive a refined contraction criterion from the spectral radius of a positive operator that majorizes the propagation of policy errors, and show that its large-horizon limit agrees with the horizon-independent contraction factor. Finally, we provide an explicit error bound between discounted and undiscounted finite-horizon regularized equilibria.
\end{abstract}

\section{Introduction}
\subsection{Background and Motivation}
Mean-field games (MFGs) provide a powerful framework for approximating solutions to multi-agent problems with identical agents. The theory of MFGs was introduced independently by Lasry and Lions in \cite{LaLi07}, who coined the term ``mean-field games'', and by Huang, Malham\'{e}, and Caines in \cite{HuMaCa06}, who formulated the problem from the viewpoint of stochastic dynamic games. These foundational works study continuous-time noncooperative differential games with a large but finite number of agents, where the influence of each individual agent on the population becomes asymptotically negligible as the number of agents grows. In this paper, we focus on finite-state MFGs in which the state of the
representative agent evolves according to a continuous-time Markov process.

MFGs reduce multi-agent problems with identical agents to effective single-agent problems and provide approximate solutions to general Markov games, for which the existence of a solution is not always guaranteed \cite{SaBaRaSIAM}. Motivated by the growing interest in multi-agent systems, computational methods for MFGs have also attracted significant attention \cite{guo2019learning}. However, MFG dynamics generally do not satisfy any natural contraction property. In fact, computing equilibria in discrete-time MFGs under mere Lipschitz conditions is PPAD-complete \cite{yardim2024mean}, suggesting that the problem is computationally intractable in general \cite{papadimitriou1994complexity}. This makes it important to identify structural conditions under which mean-field equilibria (MFE) can be computed efficiently.

There are two main setups used for computational methods in the current literature of MFGs: contraction \cite{eich2026approximately} and monotonicity \cite{dayanikli2025machine}. Under monotonicity assumptions, it is possible to obtain tractable iterative algorithms without imposing any restriction on the horizon length or the size of Lipschitz parameters \cite{carmona2018probabilistic, dayanikli2025machine}, and under further Lipschitz parameter independent structural restrictions, it is possible to obtain uniqueness of the MFE  \cite{carmona2015probabilistic} \cite{saldi2023partially}. In contrast, unlike Markov decision processes, contraction arguments mainly rely on small Lipschitz constant conditions \cite{anahtarci2023q} \cite{cui2021approximately}, which is an artifact of the forward-backward dynamics that governs the MFG. In the continuous-time case, the contraction condition also depends on the horizon length \cite{carmona2015probabilistic}.

To develop tractable algorithms for continuous-time MFGs within a learning framework, in this work we will focus on contraction-based arguments. A major challenge, however, is that in continuous time the contraction condition typically depends on the time horizon, which limits its applicability for long-horizon problems \cite{eich2026approximately} \cite{carmona2015probabilistic}. For this purpose, we study discounted MFGs instead. Besides being a well-posedness tool, we note that discounting is also commonly used in finite-horizon stochastic control problem models \cite{huangoptimal}. In this work, we revisit the dependence of contraction condition on the horizon length and show that, under mild regularity assumptions on the system components, regularized continuous-time MFGs admit a horizon-independent contraction condition whenever the discount rate is sufficiently large.

As a byproduct of our horizon-independent contraction conditions, we also obtain finite-time error bounds between infinite-horizon MFE and finite-horizon MFE for continuous-time MFGs. In general, to obtain an infinite-horizon non-stationary MFE, one has to solve an infinite-horizon dynamic programming problem, which is generally computationally intractable. Our finite-to-infinite horizon error bound therefore provides a tractable way to approximate infinite-horizon MFE. Furthermore, in practice one often seeks solutions to finite-horizon MFGs with a long but finite horizon, while instead using infinite-horizon solutions because they are structurally simpler. One example is the renewable energy model formulated as a mean-field type game in \cite[Subsection 15.7.1]{bacsar2026mean}. Our results also quantify when an infinite-horizon equilibrium provides an accurate approximation of the corresponding finite-horizon MFE with sufficiently long horizon.

\subsection{Literature Review}
\subsubsection{Computational Methods for MFE}

Computation of MFE has been studied through several complementary approaches, including convex optimization, regularization, learning-based methods, and structured control models. In discrete time, semi-definite programming and linear programming formulations have been developed to compute or approximate MFE under suitable structural assumptions \cite{guo2024mf,saldi2026linear}. Another line of work studies monotone MFGs with regularization, where entropy or other regularizers improve stability and enable convergence guarantees for iterative equilibrium computation \cite{zhang2023learning,dong2025last}. Related contraction-based ideas have also appeared in continuous-time regularized MFGs \cite{eich2026approximately}. In parallel, model-based learning methods for discrete-time MFGs have provided sample-complexity and statistical guarantees for computing equilibria when the model is unknown or estimated from data \cite{huang2024statistical,pasztor2021efficient}. For continuous-time structured models, especially linear-quadratic MFGs (LQMFGs), equilibrium computation can often be reduced to solving Riccati-type equations or entropy-regularized control problems \cite{guo2022entropy}. Our work complements previous studies on continuous-time discounted MFGs by establishing contraction and finite-to-infinite horizon convergence rates, providing theoretical guarantees that can support the computation and approximation of long-horizon MFE.

\subsubsection{Contraction Conditions for MFGs}
Contraction conditions for MFGs often require small Lipschitz coefficients for the system components as well as access to a Lipschitz continuous optimal policy \cite{guo2019learning,subramanian2019reinforcement,anahtarci2020value, eich2026approximately}. In the current literature, the most common way to satisfy these restrictions in finite-state and finite-action models is to  perturb the system objective with a \emph{regularizer} \cite{anahtarci2023q,cui2021approximately}, which causes a deviation from the true equilibria to obtain a Lipschitz continuous optimal policy. Despite the available prior work on stationary and finite-horizon non-stationary MFE, there is no fixed-point iteration algorithm studied in the case of infinite-horizon non-stationary MFG \cite{lauriere2022learning}. In the continuous-time case, the contraction conditions for the probabilistic formulation of MFGs are also known to depend on the horizon length \cite{angiuli2019cemracs}. In contrast, we show that in the presence of a sufficiently large discount rate, there exists a horizon-independent contraction condition for finite-horizon MFGs. However, the restriction on the discount rate is an artifact of the continuous-time setting and does not arise in the discrete-time case studied in \cite{aydin2025approximation}.
\subsubsection{Robustness and Stability of MFE}

Robustness and stability have become central themes in MFGs. In continuous time, robustness has been studied through PDE-based formulations of the coupled Hamilton--Jacobi--Bellman (HJB) and Fokker--Planck system \cite{bauso2016robust}, while related stability questions have been analyzed in structured models such as LQMFGs \cite{zaman2024robust} and Stackelberg-type MFGs \cite{guo2022optimization}. In discrete time, recent work examined the finite-horizon approximation of infinite-horizon equilibria \cite{aydin2025approximation2,aydin2025approximation}, and robust Markovian MFGs under ambiguity or model uncertainty \cite{langner2024markov,liang2026mean,lauriere2025robust}. It is shown in \cite{aydin2025approximation2,aydin2025approximation} that one needs stronger conditions than horizon-independent contraction condition to establish a rate of convergence between finite-horizon and infinite-horizon MFE in discrete time. In contrast, we show that no stronger condition is needed in continuous time, i.e., the horizon-independent contraction found for finite-horizon MFGs matches the corresponding infinite-horizon
non-stationary MFGs.

\subsection{Contributions}
\begin{enumerate}
    \item In Section \ref{sect:contraction}, we obtain a horizon-independent contraction condition for finite-horizon continuous-time MFGs whose state evolves according to a continuous-time Markov kernel (Theorem \ref{thrm:main}). For any given horizon length \(T\), this result amounts to representing MFE as a fixed point of a dynamical system and finding a positive operator (that depends on \(T\)) that can be used as a majorization (Lemma \ref{lem:finite-horizon-majorization}). By using the theory of positive operators on Banach spaces, we then find tight estimates for the spectral radius of this positive 
    operator as the time horizon \(T \to \infty\) (Theorem \ref{thrm:sauce2}). The proof of Theorem \ref{thrm:sauce2} also provides a quantitative method to calculate these spectral radii (Corollary \ref{cor:3.13}), which amounts to \(T\)-dependent contraction condition. Finally, we show that Bielecki metrics (exponentially weighted metrics) can be used to obtain contraction under the obtained asymptotes.

    \item In Section \ref{sect:infinite-error}, we establish contraction conditions for infinite-horizon non-stationary MFGs (Theorem \ref{thm:infty-contraction}). Through these results we demonstrate that, the horizon-independent contraction condition we have for finite-horizon MFGs in Theorem \ref{thrm:main} is the same as that for infinite-horizon MFGs. Under the contraction of the finite-horizon system, we also study the convergence rate of finite-horizon MFE to infinite-horizon non-stationary MFE.
    
    \item In Section \ref{sect:undisc}, we establish error bounds between finite-horizon discounted and undiscounted MFE (Proposition \ref{prop:map-gap}). Since finite-horizon MFGs usually admit horizon-dependent contraction conditions, our results provide insight in obtaining approximate undiscounted MFE by using discounted MFE.
\end{enumerate}

\section{Preliminaries and Notation}\label{sect:prelim}

\subsection{Finite-horizon continuous-time finite-state MFGs}
A finite-horizon continuous-time MFG with finite state and action spaces is specified by the tuple
\((\X,\A,c,\Lambda,g,\mu_0,T),\)
where:
\begin{enumerate}
    \item \(\X\) is a finite state space.
    \item \(\A\) is a finite action space.
    \item \(c:\X\times\A\times\Pcal(\X)\to\R\) is the running cost. 
    \item \(\Lambda:\X\times\X\times\A\times\Pcal(\X)\to \R\) is the controlled transition-rate kernel.
    \item \(g:\X\to\R\) is the terminal cost.
    \item \(\mu_0\in\Pcal(\X)\) is the initial state-measure.
    \item \(T<\infty\) is the horizon length.
\end{enumerate}
Here, for a finite set \(E\), by \( \Pcal(E)\) we denote the unit simplex over \(E\).
In the \(N\)-player game, under a policy flow \(\pmb \pi = (\pi_t)_{t\in [0,T]}\), we suppose that the state process of player \(i\) follows a continuous-time Markov chain of the form
\[\mathbb P(x^i_{t+h}=x'|x^i_t=x,\mu_t^N=\nu) =\delta_{xx'}+h\sum_{a \in A} \Lambda(x,x',a,\nu)\pi_t(a|x)+o(h).\]
In the mean-field regime (\(N\to \infty\)), the behavior of the population can be characterized from the lens of a single representative agent. In particular, the representative agent's state dynamics evolve according to the master equation
\begin{equation}\label{eq:forward}
    \begin{aligned}
  \dot{\tilde \mu}_t(x)
  &=
  \sum_{z\in\X}\sum_{a\in\A}
  \Lambda(z,x,a,\mu_t)\pi_t(a\mid z)\tilde \mu_t(z)
  \\&=:F^{\pmb \pi}(t,\tilde \mu_t)(x),
  \qquad \tilde \mu_0 = \mu_0\ .
\end{aligned}
\end{equation}

The quantity \(\Lambda(z,x,a,\mu)\) is the instantaneous transition rate of the controlled continuous-time Markov chain from state \(z\) to state \(x\) when the action is \(a\) and the population distribution flow is \(\pmb \mu=(\mu_t)_t\).  The flow \(\pmb \mu\) describes the population state distribution over time and will be referred to as the \emph{mean-field flow}.

For a fixed mean-field flow, the representative-agent problem in a MFG reduces to a single-agent control problem in which the population distribution is treated as an exogenous input. The representative agent's objective is to choose an admissible Markov policy \(\hat{\pmb \pi}\) \cite{eich2026approximately}, where \( \hat\pi_t : X \to \Pcal(\A)\) is a stochastic kernel, such that starting from any \((t,x)\in [0,T] \times X\), the finite-horizon cost
\begin{equation}\label{eq:classical-cost}
\begin{aligned}  
  J^{0,T}_0(t,x;\hat{\pmb \pi},\pmb \mu)
  :=
  \mathbb E_{t,x}^{\hat{\pmb \pi}}
  \left[
      \int_t^T
      \sum_{a\in\A}
      \hat\pi_s(a\mid X_s)c(X_s,a,\mu_s)\dd s
      +
      g(X_T)
  \right]
\end{aligned}
\end{equation}
is minimized over all admissible policies. 

The full MFG problem, however, additionally requires that the prescribed mean-field flow coincides with the one generated by the representative agent's optimal policy through \eqref{eq:forward}. Accordingly, an MFE consists of a pair \((\pmb \pi,\pmb \mu)\) satisfying both individual optimality and population consistency.

We first record the unregularized MFE, which serves as the reference point for the discounted regularized formulation introduced in subsection~\ref{subsec:dr-MFE}.

\begin{definition}\label{def:classical-mfe}
A pair \((\pmb \pi^T,\pmb\mu^T)\) is called a \emph{mean-field
equilibrium} (MFE) for the MFG \((\X,\A,c,\Lambda,g,\mu_0,T),\) if it satisfies the following \emph{optimality} and \emph{consistency} conditions:
\begin{enumerate}
    \item \textbf{Optimality:} for every \(t\in[0,T]\) and \(x\in\X\),
    \(
        J^{0,T}_0(t,x;\pmb \pi^T,\pmb\mu^T)
        =
        \inf_{\hat{\pmb \pi}\in\Pi_T}J^{0,T}_0(t,x;\hat{\pmb \pi},\pmb\mu^T).
    \)

    \item \textbf{Consistency:} \(\pmb \mu^T=(\mu^T_t)_{t\in[0,T]}\) is the mean-field flow induced by
    \(\pmb \pi^T\), i.e.,
    \(
      \dot\mu_t^T(x)
      =
      F^{\pmb \pi^T}(t,\mu^T_t)(x),
    \)
    where
    \(
      \mu_0^T=\mu_0.
    \)
\end{enumerate}
\end{definition}

We next introduce the norms and seminorms used throughout the analysis. For any finite set \(S\) and
any function \(f:S\to\R\), we define
\[
    \|f\|_1:=\sum_{s\in S}|f(s)|,
    \qquad
    \|f\|_\infty:=\max_{s\in S}|f(s)|,
    \qquad
    \spn(f):=\max_{s\in S}f(s)-\min_{s\in S}f(s).
\]

The following assumption ensures that the controlled generator is well defined and that the model coefficients depend Lipschitz continuously on the mean-field argument.
\begin{assumption}\label{ass:basic-reg}
The transition-rate kernel \(\Lambda\) satisfies the generator conditions
\begin{equation}\label{eq:trans}
    \Lambda(x,y,a,\mu)\ge0 \quad (y\ne x),
    \qquad
    \sum_{y\in\X}\Lambda(x,y,a,\mu)=0.
\end{equation}
Throughout, we exclude the degenerate case in which all off-diagonal
jump rates vanish. Moreover, \(c\) and \(\Lambda\) are Lipschitz continuous in the mean-field argument, uniformly in \((x,a)\). That is, there exist constants \(L_c,L_\Lambda<\infty\) such that the following inequalities hold
for all \(x\in\X\), \(a\in\A\), and \(\mu,\nu\in\Pcal(\X)\),
\[
    |c(x,a,\mu)-c(x,a,\nu)|
    \le L_c\norm{\mu-\nu}_1,
\]
and
\[
    \frac12\sum_{y\in\X}
    |\Lambda(x,y,a,\mu)-\Lambda(x,y,a,\nu)|
    \le L_\Lambda\norm{\mu-\nu}_1.
\]
\end{assumption}
Note that \eqref{eq:trans}, together with \(\mu_0\in\Pcal(\X)\), implies that the solution of the master equation \eqref{eq:forward} is indeed a flow of probability measures. Since \(\X\) and \(\A\) are finite and \(\Pcal(\X)\) is a compact subset of the finite-dimensional space \(\mathbb R^{\X}\), the
Lipschitz continuity assumptions above imply that \(c, g,\) and \(\Lambda\) are bounded. In connection with this, we define the following constants:
\begin{equation}\label{eq:constants}
    M_c:=\sup_{x,a,\mu}|c(x,a,\mu)|,
    \qquad
    M_g:=\sup_{x}|g(x)|,
    \qquad
    B_\Lambda:=
    \sup_{x,a,\mu}\sum_{y\in\X}|\Lambda(x,y,a,\mu)|.
\end{equation}
The condition in \eqref{eq:trans}
then implies that \(B_\Lambda>0\).

\subsection{Function spaces and weighted metrics}

We now introduce the spaces on which the fixed-point argument will be carried out. We equip these spaces with exponentially weighted essential-supremum metrics, which will allow us to control the propagation of perturbations over time.

Throughout the paper, we let \(L^\infty([0,T];E)\) denote the Banach space of Lebesgue measurable and essentially bounded functions \(f:[0,T]\to E\), endowed with the essential-supremum norm
\[
    \norm{f}_{L^\infty}
    :=
    \esssup_{0\le t\le T}\norm{f(t)}.
\]

Fix \(T<\infty\) and \(\eta\ge0\). 
We define the admissible policy space by:
\[
\Pi_T
:=
\left\{
\pmb \pi=(\pi_t)_{\in[0,T]}\in L^\infty([0,T];\R^{\X\times\A}):
\pi_t(\cdot\mid x)\in\Pcal(\A)
\text{ for Lebesgue-a.e. }t\in[0,T]\text{ and all }x\in\X
\right\}.
\]
For \(\pmb \pi,\pmb \pi'\in\Pi_T\), we define the point-wise policy distance
\begin{equation}\label{eq:ptswisepolicy}
  \norm{\pi_t-\pi'_t}_{\Pi}
  :=
  \max_{x\in\X}
  \sum_{a\in\A}
  |\pi_t(a\mid x)-\pi'_t(a\mid x)|.
\end{equation}
The Bielecki metric constructed from \eqref{eq:ptswisepolicy} is an exponentially weighted uniform metric on
\(\Pi_T\) given by
\begin{equation}\label{eq:policy-metric}
  d_{\eta,\Pi}^T(\pmb \pi,\pmb \pi')
  :=
  \esssup_{0\le t\le T}
  e^{-\eta t}\norm{\pi_t-\pi'_t}_{\Pi}.
\end{equation}
For \(\eta>0\), this is a Bielecki metric, whereas \(\eta=0\) recovers the usual essential-supremum metric.

The space \((\Pi_T,d_{\eta,\Pi}^T)\) is complete. Indeed, the simplex \(\Pcal(\A)\) is closed, and thus
\(\Pi_T\) is a closed subset of \(L^\infty([0,T];\R^{\X\times\A})\). Moreover, since
\(\X\) and \(\A\) are finite, the row-wise \(\ell^1\) norm
\(\max_{x\in\X}\sum_{a\in\A}|\pi_t(a\mid x)| \) is equivalent to the uniform norm on \(\R^{\X\times\A}\). On the
finite interval \([0,T]\), the weighted metric \(d_{\eta,\Pi}^T\) is therefore strongly equivalent to the standard \(L^\infty\)-metric. Completeness follows from the completeness of \(L^\infty([0,T];\R^{\X\times\A})\).

We next introduce the mean-field flow space:
\[
\mathcal M_T
:=
\left\{
 \pmb \mu=(\mu_t)_{t\in[0,T]}\in L^\infty([0,T];\R^\X):
\mu_t\in\Pcal(\X)\ \text{for Lebesgue-a.e. }t\in[0,T]
\right\}.
\]
Given an initial distribution \(\mu_0\in\Pcal(\X)\), we define the space of admissible mean-field flows by
\[
\mathcal M_T(\mu_0)
:=
\left\{
 \pmb \mu\in\mathcal M_T:
\pmb \mu \text{ is absolutely continuous on } [0,T]
\text{ and }
\mu(0)=\mu_0
\right\}.
\]
Under Assumption~\ref{ass:basic-reg}, for each admissible policy \(\pmb \pi\), the
map \(t\mapsto F^{\pmb \pi}(t,\mu)\) is Lebesgue measurable and the map
\(\mu\mapsto F^{\pmb \pi}(t,\mu)\) is Lipschitz uniformly in \(t\). Hence by the 
Carathéodory ODE theorem, there exists a unique absolutely continuous solution of
\eqref{eq:forward}.  The generator structure preserves the simplex: total mass
is preserved because each generator row sums to zero, and non-negativity is
preserved because \(F^{\pmb \pi}(t,\mu)\) is inward-pointing on the boundary of \(\Pcal(\X)\).

We also define the value-function and \(Q\)-function flow spaces as follows:
\[
\mathcal V_T:=L^\infty([0,T];\R^\X),
\qquad
\mathcal Q_T:=L^\infty([0,T];\R^{\X\times\A}).
\]
For \(\pmb \mu,\pmb \nu\in\mathcal M_T\), \(\pmb V,\pmb W\in\mathcal V_T\), and
\( \pmb Q, \pmb Q'\in\mathcal Q_T\), we define
\begin{align}
  d_{\eta,1}^T(\pmb \mu,\pmb \nu)
  &:=
  \esssup_{0\le t\le T}e^{-\eta t}\norm{\mu_t-\nu_t}_1,
  \label{eq:mu-metric}\\
  d_{\eta,V}^T(\pmb V,\pmb W)
  &:=
  \esssup_{0\le t\le T}e^{-\eta t}\norm{V_t-W_t}_\infty,
  \label{eq:value-metric}\\
  d_{\eta,Q}^T(\pmb Q,\pmb Q')
  &:=
  \esssup_{0\le t\le T}e^{-\eta t}\norm{Q_t-Q'_t}_\infty.
\end{align}
Accordingly, we work with the metric spaces
\(
    (\Pi_T, d_{\eta,\Pi}^T),
\)
\(
    (\mathcal M_T(\mu_0),d_{\eta,1}^T),
\)
\(
    (\mathcal V_T,d_{\eta,V}^T),
\)
\(
    (\mathcal Q_T,d_{\eta,Q}^T).
\)

On a fixed finite interval, the weighted metrics are equivalent to the unweighted metrics and therefore induce the same convergence and completeness properties. The parameter \(\eta\) is introduced to derive Lipschitz and contraction estimates that are uniform in the horizon length.

\subsection{Discounted regularized best response}

Relative to the classical formulation, we introduce two features that are essential for the contraction analysis. First, we incorporate a discount rate \(\delta>0\) into the objective \eqref{eq:classical-cost}. This discounting suppresses the accumulation of future perturbations in the backward HJB equation. Second, we introduce entropy regularization with scaling parameter \(\alpha>0\).  This makes the best response unique and smooth, replacing a possibly set-valued minimizer by a \emph{soft-min} policy at the cost of deviation from the unregularized equilibria \cite{guo2019learning}.
Throughout the discounted regularized analysis, we fix \(\delta>0\) and \(\alpha>0\).
For \(q\in\Pcal(\A)\), we define the regularization function as
\[
    \Omega(q):=\sum_{a\in\A}q(a)\log q(a),
\]
with the convention \(0\log0=0\).  
Since \(\Omega(q)=-H(q)\), where \(H(q)\) is the Shannon entropy,
the term \(\alpha\Omega(q)\) by which the objective function is perturbed (regularized) favors higher-entropy action distributions under the cost-minimization convention. 

Fix \(\pmb \mu\in\mathcal M_T(\mu_0)\). For \((t,x)\in[0,T]\times\X\), the discounted entropy-regularized cost associated with an admissible policy \(\hat {\pmb \pi}\) is defined by

\begin{equation}\label{eq:discounted-cost}
\begin{aligned}
J^{\delta,T}_\alpha
(t,x;\hat{\pmb\pi},\pmb\mu)
:=
\mathbb E^{\hat{\pmb\pi}}_{t,x}\bigg[&
\int_t^T e^{-\delta(s-t)}
\bigg(
\sum_{a\in\A}
\hat\pi_s(a\mid X_s)c(X_s,a,\mu_s)
+
\alpha\Omega\bigl(\hat\pi_s(\cdot\mid X_s)\bigr)
\bigg)\dd s \\
&\qquad\qquad
+
e^{-\delta(T-t)}g(X_T)
\bigg].
\end{aligned}
\end{equation}
Under Assumption~\ref{ass:basic-reg}, the cost
\(J^{\delta,T}_\alpha(t,x;\hat{\pmb \pi},\pmb \mu)\) is finite for every
\((t,x,\hat{\pmb \pi},\pmb \mu)\). The discounted regularized value function is
\begin{equation}\label{eq:discounted-value}
  V_t^{\delta,T,\pmb \mu}(x)
  :=
  \inf_{\hat{\pmb \pi}\in\Pi_T}J^{\delta,T}_\alpha(t,x;\hat{\pmb \pi},\pmb \mu).
\end{equation}

We collect the resulting discounted regularized model components in the following notation.
\begin{definition}
    Given a MFG \((\X,\A,c,\Lambda,g,\mu_0,T)\), we denote the corresponding \(\delta\)-discounted \(\alpha\)-regularized MFG by \(\mathrm{MFG}_{\delta,T,\alpha}=(\X,\A,c,\delta,\alpha,\Omega,\Lambda,g,\mu_0,T)\).
\end{definition}
We next characterize the representative agent's best response to a fixed mean-field flow \(\pmb\mu\). 
For \(v\in\R^\X\), define the Hamiltonian
\begin{equation}\label{eq:ham}
\begin{aligned}
H_t^{\pmb\mu}(x,v)
:=
\inf_{q\in\Pcal(\A)}
\bigg\{
\sum_{a\in\A}q(a)
\left[
c(x,a,\mu_t)
+
\sum_{y\in\X}
\Lambda(x,y,a,\mu_t)v(y)
\right]
+
\alpha\Omega(q)
\bigg\}.
\end{aligned}
\end{equation}
By the dynamic programming principle for finite-state continuous-time Markov decision processes
\cite{shreve1979universally,fleming2006controlled}, the value function flow \(\pmb V^{\delta,T,\pmb\mu} = \bigl(V_t^{\delta,T,\pmb\mu}\bigr)_{t\in[0,T]}\) defined in \eqref{eq:discounted-value}, 
satisfies the following backward HJB equation. For a.e.
\(t\in[0,T]\) and every \(x\in\X\),
\begin{equation}\label{eq:HJB}
  -\dot V_t^{\delta,T,\pmb\mu}(x)
  =
  H_t^{\pmb\mu}
  \bigl(x,V_t^{\delta,T,\pmb\mu}\bigr)
  -
  \delta V_t^{\delta,T,\pmb\mu}(x),
  \qquad
  V_T^{\delta,T,\pmb\mu}(x)=g(x).
\end{equation}

The entropy regularization also makes the pointwise minimization in
\eqref{eq:ham} explicit. 
For fixed \((t,x,\pmb \mu,v)\), collect the
action-dependent terms in the Hamiltonian by setting
\(
    Z_a
    :=
    c(x,a,\mu_t)
    + 
    \sum_{y\in\X} \Lambda(x,y,a,\mu_t) v(y), 
     a \in \A.
\) By the Gibbs variational principle
\cite[Theorem~3.4]{wainwright2008graphical}, the Hamiltonian admits the soft-min representation
\(
  H_t^{\pmb\mu}(x,v)
  =
  -\alpha
  \log\left(
      \sum_{a\in\A}\exp(-Z_a/\alpha)
  \right),
\)
Moreover, the unique minimizer of \eqref{eq:ham} is
\(
  q^*(a;Z)
    =
    \frac{\exp(-Z_a/\alpha)}
         {\sum_{b\in\A}\exp(-Z_b/\alpha)}, \; a\in\A.
\)
Setting \(v=V_t^{\delta,T,\pmb\mu}\), the quantity \(Z_a\) becomes the discounted \(Q\)-function
\begin{equation}\label{eq:Qdelta}
  Q_t^{\delta,T,\pmb \mu}(x,a)
  :=
  c(x,a,\mu_t)
+\sum_{y\in\X}\Lambda(x,y,a,\mu_t)V_t^{\delta,T,\pmb \mu}(y).
\end{equation}
Consequently, the representative agent's regularized best response to \(\pmb\mu\) is
\begin{equation}\label{eq:softmin-Q}
  q_t^{*,T,\pmb\mu}(a\mid x)
  =
  \frac{
      \exp\left(
          -Q_t^{\delta,T,\pmb\mu}(x,a)/\alpha
      \right)}
       {\displaystyle
        \sum_{b\in\A}
        \exp\left(
            -Q_t^{\delta,T,\pmb\mu}(x,b)/\alpha
        \right)}.
\end{equation}

We conclude this subsection by recording the well-posedness and measurability of the preceding constructions. 
Under Assumption~\ref{ass:basic-reg}, the maps
\(\mu\mapsto c(x,a,\mu),\)
\(\mu\mapsto \Lambda(x,y,a,\mu)\)
are Lipschitz and hence Borel measurable.  Therefore, for every
\(\pmb \mu\in\mathcal M_T(\mu_0)\), the maps
\(t\mapsto c(x,a,\mu_t),\) and \(t\mapsto \Lambda(x,y,a,\mu_t)\) are Lebesgue measurable for every \(x,y\in\X\) and \(a\in\A\). 
Since \(\X\) and \(\A\) are finite, the map \(v\mapsto H_t^{\pmb \mu}(x,v)\) is locally
Lipschitz, uniformly in \(t\) on bounded sets.  
 Carathéodory's existence--uniqueness theorem therefore yields a unique absolutely continuous solution of the backward ODE \eqref{eq:HJB}
\cite[Ch.~1, Sec.~1, Thms.~1.1--1.2, pp.~43--47]
{coddington1955theory}.
The dynamic programming principle and the verification argument identify this solution with the value function defined in \eqref{eq:discounted-value} \cite[Sec.~I.4--I.5, pp.~27--28 and Theorem~5.2, pp.~32--33] {fleming2006controlled}.
Since this solution is absolutely continuous on the compact interval \([0,T]\), it is measurable and bounded. Therefore,
\(\pmb V^{\delta,T,\pmb\mu}\in\mathcal V_T.\)
By \eqref{eq:Qdelta} and the measurability and boundedness established above, \( \pmb Q^{\delta,T,\pmb\mu}\in\mathcal Q_T.\)
Finally, the finite-dimensional soft-min map is smooth and maps \(\R^\A\) into \(\Pcal(\A)\). Therefore, \eqref{eq:softmin-Q} defines
a measurable policy flow satisfying \(\pmb q^{*,T,\pmb\mu}\in\Pi_T.\)

\subsection{The MFE fixed-point operator}
\label{subsec:dr-MFE}
We now combine the forward population dynamics with the discounted
regularized best response to formulate the MFE problem as a fixed-point problem. The construction consists of three
maps: a policy flow generates a mean-field flow, the mean-field flow determines a discounted \(Q\)-function, and the \(Q\)-function induces a new policy. 
This reformulation allows us to prove existence and uniqueness of the MFE via a contraction argument. Furthermore, this contraction property yields a convergent iterative method for computing the
MFE.

The first operator maps a policy flow to the mean-field flow it induces.
\begin{definition}\label{def:GammaM}
We define the \emph{forward mean-field operator}
\(\Gamma_M^T:\Pi_T\to\mathcal M_T(\mu_0)\) by setting
\(\Gamma_M^T(\pmb \pi)=\pmb \mu\), where \(\pmb \mu\) is the unique absolutely
continuous solution of the forward equation \eqref{eq:forward}, namely,
\(
  \dot\mu_t(x)
  =
  \sum_{z\in\X}\sum_{a\in\A}
  \Lambda(z,x,a,\mu_t)\pi_t(a\mid z)\mu_t(z), \;
  \mu(0)=\mu_0,
\)
for a.e. \(t\in[0,T]\) and every \(x\in\X\).
\end{definition}
The second operator maps a mean-field flow to the discounted \(Q\)-function flow.

\begin{definition}\label{def:GammaQ}
We define the \emph{discounted \(Q\)-function operator}
\(\Gamma_Q^{\delta , T}:\mathcal M_T(\mu_0)\to\mathcal Q_T\) by setting
\(\Gamma_Q^{\delta , T}(\pmb \mu)=\pmb Q^{\delta,T,\pmb \mu}\), where
, by \eqref{eq:Qdelta},
\(
  Q_t^{\delta,T,\pmb \mu}(x,a)
  =
  c(x,a,\mu_t)
  +
  \sum_{y\in\X}
  \Lambda(x,y,a,\mu_t)V_t^{\delta,T,\pmb \mu}(y), 
\) 
for a.e. \(t\in[0,T]\) and every \((x,a)\in\X\times\A\), with \( \pmb V^{\delta,T,\pmb \mu}\) being the solution of HJB equation \eqref{eq:HJB}.
\end{definition}

The third operator maps a \(Q\)-function flow to the regularized policy flow.
\begin{definition}\label{def:GammaPi}
We define the \emph{soft-min policy operator}
\(\Gamma_\Pi^T:\mathcal Q_T\to\Pi_T\) by the point-wise rule \eqref{eq:softmin-Q}, namely, for each \(\pmb Q\in\mathcal Q_T\),
\(
  \bigl(\Gamma_\Pi^T(\pmb Q)\bigr)_t(a\mid x)
  :=
  \frac{\exp(-Q_t(x,a)/\alpha)}
       {\sum_{b\in\A}\exp(-Q_t(x,b)/\alpha)},
\)
for a.e. \(t\in[0,T]\) and every \((x,a)\in\X\times\A\).
\end{definition}
Composing these three mappings gives the policy update underlying the MFE fixed-point formulation.
\begin{definition}\label{def:PhiDelta}
We define the \emph{finite-horizon MFE operator} of \(\mathrm{MFG}_{\delta,T,\alpha}\) as the composition of the three operators introduced above:
\begin{equation*}\label{eq:Phi-delta}
  \Phi_\delta^T
  :=
  \Gamma_\Pi^T\circ\Gamma_Q^{\delta , T}\circ\Gamma_M^T
  :
  \Pi_T\to\Pi_T.
\end{equation*}
\end{definition}

For any candidate policy flow \(\pmb\pi\), the updated policy
\(\Phi_\delta^T(\pmb\pi)\) is therefore the representative agent's
discounted regularized best response to the mean-field flow generated by
\(\pmb\pi\).

\begin{definition}\label{def:regularized-mfe}
A pair
\(
    (\pmb \pi^{\delta,T},\pmb\mu^{\delta,T})
    \in
    \Pi_T\times\mathcal M_T(\mu_0)
\)
is called a finite-horizon \((\delta,\alpha)\)-regularized MFE of \(\mathrm{MFG}_{\delta,T,\alpha}\) if 
    \(
        \pmb \mu^{\delta,T}
        =
        \Gamma_M^T(\pmb \pi^{\delta,T})
    \)
    and
    \(  \pmb \pi^{\delta,T}=\Gamma_\Pi^T\bigl(\Gamma_Q^{\delta , T}(\pmb \mu^{\delta,T})\bigr).
    \)
Equivalently,
\(
    \pmb \pi^{\delta,T}
    =
    \Phi_\delta^T(\pmb \pi^{\delta,T}),
\)
and
\(
    \pmb \mu^{\delta,T}
    =
    \Gamma_M^T(\pmb \pi^{\delta,T}).
\)
\end{definition}

Thus, the fixed points of \(\Phi_\delta^T\) are in one-to-one
correspondence with the finite-horizon
\((\delta,\alpha)\)-regularized MFEs.
Hence, once \(\Phi_\delta^T\) is shown to be a contraction on a complete policy space, Banach's fixed-point theorem yields existence and uniqueness of the MFE, as well as global convergence of the associated best-response Picard iteration \(\pmb \pi^{k+1}=\Phi_\delta^T(\pmb \pi^k).\)

In Section \ref{sect:contraction}, we derive a finite-horizon contraction condition that is uniform in the horizon length \(T\).

\subsection{Infinite-horizon non-stationary MFGs}
\label{subsec:infinite-ns}
We now extend the preceding formulation to the infinite-horizon
non-stationary setting. We denote the resulting model by
\(
    \mathrm{MFG}_{\delta,\infty,\alpha}
    :=
    (\X,\A,c,\delta,\alpha,\Omega,\Lambda,\mu_0).
\)
It is obtained from \(\mathrm{MFG}_{\delta,T,\alpha}\) by taking \(g \equiv 0\) and \(T = \infty\). Here,
non-stationary means that the policy and mean-field flows may remain time dependent.

Fix \(\pmb\mu\in\mathcal M_\infty(\mu_0)\). For
\((t,x)\in[0,\infty)\times\X\), 
the representative agent minimizes the infinite-horizon discounted
regularized cost, given by
\begin{equation}\label{eq:infinite-discounted-cost}
\begin{aligned}
  J_{\alpha}^{\delta,\infty}(t,x;\hat {\pmb \pi},\pmb \mu)
  :=
  \mathbb E_{t,x}^{\hat{\pmb \pi}}\bigg[
      \int_t^\infty
      e^{-\delta(s-t)}
      \left(
      \sum_{a\in\A}\hat\pi_s(a\mid X_s)c(X_s,a,\mu_s)
      +
      \alpha\Omega(\hat\pi_s(\cdot\mid X_s))
      \right)\dd s
  \bigg].
\end{aligned}
\end{equation}
The corresponding value function
is denoted by \(\pmb V^{\delta,\infty,\pmb \mu} =\inf_{\hat{\pmb\pi}\in\Pi_\infty}
J_{\alpha}^{\delta,\infty}
(t,x;\hat{\pmb\pi},\pmb\mu)\). The Hamiltonian, \(Q\)-function, and
soft-min rule are defined by the same formulas as in the finite-horizon case, with \(T\) replaced by \(\infty\). The essential difference is the HJB equation. By the same
dynamic programming argument, \( \pmb V^{\delta,\infty,\pmb\mu}
    =
    \bigl(V_t^{\delta,\infty,\pmb\mu}\bigr)_{t\ge0}\) satisfies
\begin{equation}\label{eq:HJB-infty}
  -\dot V_t^{\delta,\infty,\pmb\mu}(x)
  =
  H_t^{\pmb\mu}
  \bigl(x,V_t^{\delta,\infty,\pmb\mu}\bigr)
  -
  \delta V_t^{\delta,\infty,\pmb\mu}(x),
  \qquad t\ge0.
\end{equation}
Unlike the finite-horizon HJB equation, \eqref{eq:HJB-infty} has no terminal condition. Instead, the relevant solution is selected by boundedness. Equivalently, it satisfies the integral equation
\begin{equation} \label{eq:HJB-infty-mild}
    V_t^{\delta,\infty,\pmb \mu}(x)
  =
  \int_t^\infty
  e^{-\delta(s-t)}
  H_s^{\pmb \mu}(x,V_s^{\delta,\infty,\pmb \mu})\,\dd s .
\end{equation}

Similar to the finite-horizon case, we define a Banach space
\[
L_\eta^\infty([0,\infty);E)
:=
\left\{
f:[0,\infty)\to E:
f \text{ is Lebesgue measurable and }
\norm{f}_{\eta,\infty}<\infty
\right\},
\]
where
\(
    \norm{f}_{\eta,\infty}
    :=
    \esssup_{t\ge0}e^{-\eta t}\norm{f_t},
\) and \(\eta\ge0\) , \(E\) is a finite-dimensional
Euclidean space. Indeed, the map
\(f\mapsto (t\mapsto e^{-\eta t}f_t)\) is an isometry from
\(L_\eta^\infty([0,\infty);E)\) onto the ordinary space
\(L^\infty([0,\infty);E)\), which is complete.

The spaces
\((\Pi_\infty,d_{\eta,\Pi}^{\infty})\),
\((\mathcal M_\infty(\mu_0),d_{\eta,1}^{\infty})\),
\((\mathcal V_\infty^\eta,d_{\eta,V}^{\infty})\), and
\((\mathcal Q_\infty^\eta,d_{\eta,Q}^{\infty})\) are defined as the
corresponding finite-horizon spaces with \([0,T]\) replaced by
\([0,\infty)\). 
 \((\Pi_\infty, d_{\eta,\Pi}^{\infty})\) is complete. 
Indeed, if \((\pmb \pi^n)_n\) is Cauchy under
\(d_{\eta,\Pi}^{\infty}\), then \(\{(e^{-\eta t}\pi_t^n)_t:n\in \mathbb N\}\) can be viewed as a sequence in \(L^\infty([0,\infty);\R^{\X\times\A})\), which is also Cauchy and hence converges to some
\(\bar{\pmb \pi}\). Since the weighted simplex constraints are closed, the limit satisfies
\(
    \bar\pi_t(a\mid x)\ge0, \;
    \sum_{a\in\A}\bar\pi_t(a\mid x)=e^{-\eta t}
\)
for a.e. \(t\ge0\). Therefore
\(\pi_t(a\mid x):=e^{\eta t}\bar\pi_t(a\mid x)\) defines an element of
\(\Pi_\infty\), and \(\pmb \pi^n\to\pmb \pi\) under
\(d_{\eta,\Pi}^{\infty}\). 

The operators
\(\Gamma_M^\infty,\) \(  \Gamma_{Q}^{\delta,\infty}\), \(\Gamma_\Pi^\infty\)
are defined as in their finite-horizon counterparts, with
the time domain replaced by \([0,\infty)\). We set
\(
  \Phi_\delta^\infty
  :=
  \Gamma_\Pi^\infty
  \circ
  \Gamma_{Q}^{\delta,\infty}
  \circ
  \Gamma_M^\infty
  :
  \Pi_\infty\to\Pi_\infty .
\)
An infinite-horizon non-stationary \((\delta,\alpha)\)-regularized MFE is a
pair
\(
    (\pmb \pi^{\delta,\infty},\pmb \mu^{\delta,\infty})
    \in
    \Pi_\infty\times\mathcal M_\infty(\mu_0)
\)
such that
\(
    \pmb \pi^{\delta,\infty}=\Phi_\delta^\infty(\pmb \pi^{\delta,\infty})\), \(
    \pmb \mu^{\delta,\infty}=\Gamma_M^\infty(\pmb \pi^{\delta,\infty}).
\)

\section{Horizon-Independent Contraction of the MFE Operator in the Finite-Horizon Case}\label{sect:contraction}

Let \(\mathrm{MFG}_{\delta,T,\alpha}=(\X,\A,c,\delta,\alpha,\Omega,\Lambda,g,\mu_0,T)\) be a finite-horizon \(\delta\)-discounted \(\alpha\)-regularized MFG. In our setting, when a finite-horizon MFG is not discounted, i.e., \(\delta=0\), the corresponding MFE
operator for \(\mathrm{MFG}_{\delta,T,\alpha}\) might not admit a horizon-independent contraction condition
\cite{eich2026approximately}. In contrast, in the regularized setting, the discounted formulation allows
us to work with a Bielecki metric whose exponential weight is independent of the horizon length. In particular, we prove the following theorem for \(\mathrm{MFG}_{\delta,T,\alpha}\).
\begin{theorem}\label{thrm:main}
Let \(\mathrm{MFG}_{\delta,T,\alpha}\) be a finite-horizon discounted regularized MFG. Let $M_c$ and $B_{\Lambda}$ be as in \eqref{eq:constants}.
Suppose that Assumption~\ref{ass:basic-reg} holds.  
If it holds that
\[
  \kappa_{\delta,\eta,\alpha}^T
  :=
  \frac{B_\Lambda}{\eta-B_\Lambda-2L_\Lambda}
  \cdot
  \frac1{\alpha}
  \cdot
  \left(L_c+L_\Lambda \left(\spn(g)+\frac{2M_c}{\delta-B_\Lambda}\right)\right)
  \left(
      1+\frac{B_\Lambda}{\delta-B_\Lambda-\eta}
  \right)
  <1 ,
\]
for some \(\eta\) such that \(
    B_\Lambda+2L_\Lambda<\eta<\delta-B_\Lambda,
\) then, for every finite horizon \(T<\infty\), the MFE
operator
\(
    \Phi_\delta^T:\Pi_T\to\Pi_T
\)
of \(\mathrm{MFG}_{\delta,T,\alpha}\)
is (eventually) contractive under the uniform norm. 
Consequently, for every \(T<\infty\), there exists a unique finite-horizon
\((\delta,\alpha)\)-regularized MFE
\((\pmb \pi^{\delta,T},\pmb \mu^{\delta,T})\). Moreover, for every
\(\pmb \pi^0\in\Pi_T\), the iteration
\(
    \pmb \pi^{k+1}=\Phi_\delta^T(\pmb \pi^k)
\)
converges to \(\pmb \pi^{\delta,T}\) in \(d_{\eta,\Pi}^T\), with
\(\pmb \mu^{\delta,T}=\Gamma_M^T(\pmb \pi^{\delta,T})\).
\end{theorem}


The analysis proceeds in three steps. We first establish Lipschitz bounds for the three components \(\Gamma_M^T\), \(\Gamma_Q^{\delta,T}\), and \(\Gamma_\Pi^T\) of the MFE operator.  We then construct a positive compact operator \(\mathcal A_T^\delta\) majorizing the policy error, whose spectral radius yields a refined contraction criterion for each fixed horizon.
Finally, combining the component-wise estimates under Bielecki metrics proves the horizon-independent contraction result in Theorem~\ref{thrm:main}. The subsequent spectral comparison shows that the optimized Bielecki bound coincides with the large-horizon limit of the spectral radius of \(A^{\delta}_T\). 
The spectral analysis also provides a computable characterization of the \(T\)-dependent contraction rate, although generally not in closed form.


\subsection{Lipschitz bounds for \(\Gamma_M^T, \Gamma_Q^{\delta,T},\) and \(\Gamma_{\Pi}^T\)}
We first show that the forward mean-field function \(F^{\pmb \pi}(t,\cdot)\) induces a Lipschitz continuous operator on \((\mathcal P(X),\|\cdot\|_1)\). 

We will use the notation
\(
    a_M:=B_\Lambda+2L_\Lambda
\) for the rest of the paper.
\begin{lemma}\label{lem:F-lip}
Under Assumption~\ref{ass:basic-reg}, for all admissible
policies \(\pmb \pi,\pmb \pi' \in\Pi_T\), a.e. \(t\in[0,T]\), and all
\(\mu,\nu\in\Pcal(\X)\),
\begin{align}
  \norm{F^{\pmb \pi}(t,\mu)-F^{\pmb \pi}(t,\nu)}_1
  &\le
  a_M\norm{\mu-\nu}_1,
  \label{eq:fm}\\
  \norm{F^{\pmb \pi}(t,\mu)-F^{\pmb \pi'}(t,\mu)}_1
  &\le
  B_\Lambda\norm{\pi_t-\pi'_t}_{\Pi}.
  \label{eq:fpi}
\end{align}
\end{lemma}

\begin{proof}
Adding and subtracting
\(\Lambda(z,x,a,\mu)\pi_t(a\mid z)\nu(z)\), we obtain
\[
\begin{aligned}
F^{\pmb \pi}(t,\mu)(x)-F^{\pmb \pi}(t,\nu)(x)
&=
\sum_{z,a}\Lambda(z,x,a,\mu)\pi_t(a\mid z)(\mu(z)-\nu(z))  \\
&\quad+
\sum_{z,a}
\bigl(\Lambda(z,x,a,\mu)-\Lambda(z,x,a,\nu)\bigr)
\pi_t(a\mid z)\nu(z).
\end{aligned}
\]
Taking the \(1\)-norm over \(x\), using
\(\sum_a\pi_t(a\mid z)=1\), \(\sum_z\nu(z)=1\), and Assumption~\ref{ass:basic-reg}, leads to
\[
\begin{aligned}
\norm{F^{\pmb \pi}(t,\mu)-F^{\pmb \pi}(t,\nu)}_1
&\le
B_\Lambda\norm{\mu-\nu}_1
+
2L_\Lambda\norm{\mu-\nu}_1  
=
a_M\norm{\mu-\nu}_1.
\end{aligned}
\]
This proves \eqref{eq:fm}.

Similarly, for fixed \(\mu\), since it holds that
\[
F^{\pmb \pi}(t,\mu)(x)-F^{\pmb \pi'}(t,\mu)(x)
=
\sum_{z,a}\Lambda(z,x,a,\mu)
\left(\pi_t(a\mid z)-\pi'_t(a\mid z)\right)\mu(z),
\]
we obtain
\begin{align*}
\norm{F^{\pmb \pi}(t,\mu)-F^{\pmb \pi'}(t,\mu)}_1
&\le
\sum_{z\in\X}\mu(z)
\sum_{a\in\A}
|\pi_t(a\mid z)-\pi'_t(a\mid z)|
\sum_{x\in\X}|\Lambda(z,x,a,\mu)| \\
&\le
B_\Lambda
\sum_{z\in\X}\mu(z)
\sum_{a\in\A}
|\pi_t(a\mid z)-\pi'_t(a\mid z)| \le
B_\Lambda\norm{\pi_t-\pi'_t}_{\Pi},
\end{align*}
which proves \eqref{eq:fpi}.
\end{proof}

The next lemma integrates these local estimates over time and gives a Lipschitz bound for the forward mean-field operator \(\Gamma_M^T:\pmb \pi\mapsto \pmb \mu \).
\begin{lemma}\label{lem:forward}
Let \(\pmb \mu=\Gamma_M^T(\pmb \pi)\) and \(\pmb \nu=\Gamma_M^T(\pmb \pi')\).
If \(\eta>a_M\), then
\[
  d_{\eta,1}^T(\pmb \mu,\pmb \nu)
  \le
  \frac{B_\Lambda}{\eta-a_M}
 d_{\eta,\Pi}^T(\pmb \pi,\pmb \pi').
\]
\end{lemma}
\begin{proof}
Since \(\mu_0=\nu_0\), integrating the corresponding forward equations \eqref{eq:forward} gives 
\begin{align*}
  \mu_t-\nu_t&=\int_0^t\bigl(F^{\pmb \pi}(s,\mu_s)-F^{\pmb \pi'}(s,\nu_s)\bigr)\dd s
  \\&=\int_0^t\bigl(F^{\pmb \pi}(s,\mu_s)-F^{\pmb \pi}(s,\nu_s)\bigr)ds+\int_0^t \bigl(F^{\pmb \pi}(s,\nu_s)-F^{\pmb \pi'}(s,\nu_s)\bigr)ds.
\end{align*}
Using Lemma~\ref{lem:F-lip}, we obtain
\begin{align}
  \norm{\mu_t-\nu_t}_1
  &\le \int_0^t a_M\norm{\mu_s-\nu_s}_1\dd s
     +\int_0^t B_\Lambda\norm{\pi_s-\pi'_s}_{\Pi}\dd s 
    \\& \le B_\Lambda\int_0^t e^{a_M(t-s)}\norm{\pi_s-\pi'_s}_{\Pi}\dd s, \label{eq:mu-it}
\end{align}
where the last line follows from Gr\"onwall's inequality.

Note that \eqref{eq:policy-metric} implies \(\norm{\pi_s-\pi'_s}_{\Pi}\le e^{\eta s}d_{\eta,\Pi}^T(\pmb \pi,\pmb \pi')\) for a.e. \(s\). Multiplying \eqref{eq:mu-it} with \(e^{-\eta t}\), it then follows that for a.e. $t$ we have
\begin{align}
  e^{-\eta t}\norm{\mu_t-\nu_t}_1
  &\le B_\Lambda d_{\eta,\Pi}^T(\pmb \pi,\pmb \pi')
  \int_0^t e^{-(\eta-a_M)(t-s)}\dd s \le \frac{B_\Lambda}{\eta-a_M}d_{\eta,\Pi}^T(\pmb \pi,\pmb \pi'), \label{eq:mu-it-2}
\end{align}
where the last inequality follows from the assumption \(\eta>a_M\), since under this assumption we must have
\[
  \int_0^t e^{-(\eta-a_M)(t-s)}\dd s\le \frac1{\eta-a_M}.
\]
Taking the essential supremum of \eqref{eq:mu-it-2} over \(t\in[0,T]\) proves the desired result.
\end{proof}
To estimate the Lipschitz dependence of \(\Gamma_Q^{\delta,T}\) on the
mean-field flow, we need a uniform bound on the variation of the value function across states.  The natural quantity is \(\spn(V_t)\), since the generator rows have zero total mass and therefore additive constants in \(V_t\) disappear from
\(
    \sum_{y\in\X}\Lambda(x,y,a,\mu)V_t(y).
\)
The next lemma provides a horizon-independent span bound.
\begin{lemma}\label{lem:span}
Suppose that \(\delta>B_\Lambda\).  Then, for all  \(T<\infty\), \(t\in[0,T]\), and every mean-field flow \(\pmb \mu \in \mathcal M_T(\mu_0)\),
\begin{equation}\label{eq:span-est}
  \spn(V_t^{\delta,T,\pmb \mu})
  \le S_V^\delta:=\spn(g)+\frac{2M_c}{\delta-B_\Lambda}.
\end{equation}
\end{lemma}
\begin{proof}
We first record a simple property of the soft-min map.  For
\(z=(z_a)_{a\in\A}\in\R^\A\), define
\[
    \ell_\alpha(z):=-\alpha\log\sum_{a\in\A}e^{-z_a/\alpha}.
\]
Then \(\ell_\alpha\) is \(1\)-Lipschitz with respect to the sup norm.  Indeed,
\[
    \frac{\partial \ell_\alpha(z)}{\partial z_a}
    =
    \frac{e^{-z_a/\alpha}}{\sum_{b\in\A}e^{-z_b/\alpha}},
\text{ 
and hence }
    \sum_{a\in\A}
    \left|
    \frac{\partial \ell_\alpha(z)}{\partial z_a}
    \right|
    =
    1.
\]
By the mean-value theorem,
\(
    |\ell_\alpha(z)-\ell_\alpha(z')|
    \le
    \norm{z-z'}_\infty.
\)
Therefore, for any \(x,z\in\X\) and any \(v:\X\to\R\),
\[
\begin{aligned}
  |H_t^{\pmb \mu}(x,v)-H_t^{\pmb \mu}(z,v)|
  &\le
  \sup_{a\in\A}
  \left|
      c(x,a,\mu_t)-c(z,a,\mu_t)
      +
      \sum_{y\in\X}
      \bigl(\Lambda(x,y,a,\mu_t)-\Lambda(z,y,a,\mu_t)\bigr)v(y)
  \right|\\
  &\le
  2M_c+
  \sup_{a\in\A}
  \left|
      \sum_{y\in\X}
      \bigl(\Lambda(x,y,a,\mu_t)-\Lambda(z,y,a,\mu_t)\bigr)v(y)
  \right|.
\end{aligned}
\]
Fix \(a\in\A\) and set
\( h_y:=\Lambda(x,y,a,\mu_t)-\Lambda(z,y,a,\mu_t). \)
By Assumption~\ref{ass:basic-reg}, it follows that
\(
    \sum_{y\in\X}h_y=0.
\)
Hence, for any constant \(r\in\R\),
\( \sum_{y\in\X}h_yv(y) = \sum_{y\in\X}h_y(v(y)-r). \)
Let
\(
    r:=\frac{\max_{y\in\X}v(y)+\min_{y\in\X}v(y)}{2}.
\)
Then, for all \(y \in X\), it holds that
\(
    |v(y)-r|\le 2^{-1}\spn(v).
\)
Consequently, we obtain that
\[
\left|
    \sum_{y\in\X}h_yv(y)
\right|
=
\left|
    \sum_{y\in\X}h_y(v(y)-r)
\right|
\le
\frac12\sum_{y\in\X}|h_y|\spn(v)
\le
B_\Lambda\spn(v).
\]
Thus, it holds that
\begin{equation}\label{eq:H-span-lip}
  |H_t^{\pmb \mu}(x,v)-H_t^{\pmb \mu}(z,v)|
  \le 2M_c+B_\Lambda\spn(v).
\end{equation}
The variation-of-constants formula for \eqref{eq:HJB} yields
\begin{equation}\label{eq:V-span}
  V_t^{\delta,T,\pmb \mu}(x)=e^{-\delta(T-t)}g(x)
  +\int_t^T e^{-\delta(s-t)}H_s^{\pmb \mu}(x,V_s^{\delta,T,\pmb \mu})\dd s.
\end{equation}
For all \( t \in [0,T]\), we set \(\mathfrak{s}_t:=\spn(V_t^{\delta,T,\pmb \mu}).\)
Since \(\X\) is finite, any terminal cost \(g:\X\to\R\) is bounded, \(\spn(g)<\infty\).
With this notation, \eqref{eq:H-span-lip} and \eqref{eq:V-span} imply
\begin{equation}\label{eq:gron-bc}
  \mathfrak{s}_t\le e^{-\delta(T-t)}\spn(g)
  +\int_t^T e^{-\delta(s-t)}
  \bigl(2M_c+B_\Lambda \mathfrak{s}_s\bigr)\dd s.
\end{equation}
Applying Gr\"onwall's inequality on \([t,T]\) to \eqref{eq:gron-bc} yields
\begin{align*}
  \mathfrak{s}_t
  &\le e^{-(\delta-B_\Lambda)(T-t)}\spn(g)+2M_c\int_t^T e^{-(\delta-B_\Lambda)(s-t)}\dd s\le \spn(g)+\frac{2M_c}{\delta-B_\Lambda}
  =S_V^\delta,
\end{align*}
where the second inequality follows from the assumption that \(\delta>B_\Lambda\) as we must have
\[
  e^{-(\delta-B_\Lambda)(T-t)}\le 1, \text{ and }
  \int_t^T e^{-(\delta-B_\Lambda)(s-t)}\dd s
  \le \frac{1}{\delta-B_\Lambda}.
\]
Thus, we have obtained \eqref{eq:span-est}.
\end{proof}
We next estimate how the discounted HJB solution \(\pmb V^{\delta,T,\pmb \mu}\) depends on the mean-field flow \(\pmb \mu\). 
This Lipschitz bound is the key step in proving the Lipschitz continuity of the \(Q\)-function operator
\(\Gamma_Q^{\delta,T}:\mathcal M_T(\mu_0)\to\mathcal Q_T\).
\begin{lemma}\label{lem:hjb}
Let \(\pmb \mu, \pmb \nu \in \mathcal M_T(\mu_0)\) be mean-field flows.  Assume \(\delta>B_\Lambda\), and define
\begin{equation}\label{eq:Cmu}
  C_\mu^{\delta,T}:=L_c+L_\Lambda S_V^\delta.
\end{equation}
If \(0\le\eta<\delta-B_\Lambda\), then
\begin{equation}\label{eq:V-lip}
  d_{\eta,V}^T(\pmb V^{\delta,T,\pmb \mu},\pmb V^{\delta,T,\pmb \nu})
  \le \frac{C_\mu^{\delta,T}}{\delta-B_\Lambda-\eta}d_{\eta,1}^T(\pmb \mu,\pmb \nu).
\end{equation}
\end{lemma}

\begin{proof}
Using the variation-of-constants formula for the value functions \(\pmb V^{\delta,T,\pmb \mu}\) and \(\pmb V^{\delta,T,\pmb \nu}\), and noting that their terminal conditions are both \(g\), the terminal terms cancel out and
\[
  \norm{ V^{\delta,T,\pmb \mu}_t-V^{\delta,T,\pmb \nu}_t}_\infty\le \int_t^T e^{-\delta(s-t)}
  \|Q_s^{\delta,T,\pmb \mu}-Q_s^{\delta,T,\pmb \nu}\|\dd s.
\]
Let
\(
  d_t:=
  \norm{
  V_t^{\delta,T,\pmb \mu}
  -
  V_t^{\delta,T,\pmb \nu}
  }_\infty,
\)
and
\(
  m_t:=\norm{\mu_t-\nu_t}_1 .
\)
Using \eqref{eq:Qdelta}, for each \((s,x,a)\),
\begin{align*}
|Q_s^{\delta,T,\pmb \mu}(x,a)-Q_s^{\delta,T,\pmb \nu}(x,a)|&\le \underbrace{|c(x,a,\mu_s)-c(x,a,\nu_s)|}_{(I)}
 +\underbrace{\left|\sum_y\Lambda(x,y,a,\mu_s)( V^{\delta,T,\pmb \mu}_s(y)-V^{\delta,T,\pmb \nu}_s(y))\right|}_{(II)} \\
&\quad +\underbrace{\left|\sum_y\bigl(\Lambda(x,y,a,\mu_s)-\Lambda(x,y,a,\nu_s)\bigr)V^{\delta,T,\pmb \nu}_s(y)\right|}_{(III)}.
\end{align*}
The Assumption \ref{ass:basic-reg} implies that the term \((I)\) is bounded above by \(L_cm_s\) and the term \((II)\) is bounded above by \(B_\Lambda d_s\).  For the term \((III)\), Assumption \ref{ass:basic-reg} implies 
\(
  \sum_{y\in\X}
  \bigl(\Lambda(x,y,a,\mu_s)-\Lambda(x,y,a,\nu_s)\bigr)=0.
\) Hence, for any constant \(r\in\R\),
\(
\sum_y
\bigl(\Lambda(x,y,a,\mu_s)-\Lambda(x,y,a,\nu_s)\bigr)
V_s^{\delta,T,\pmb \nu}(y)
=
\sum_y
\bigl(\Lambda(x,y,a,\mu_s)-\Lambda(x,y,a,\nu_s)\bigr)
\bigl(V_s^{\delta,T,\pmb \nu}(y)-r\bigr).
\)
Choosing
\(
    r=\frac{\max_y V_s^{\delta,T,\pmb \nu}(y)+\min_y V_s^{\delta,T,\pmb \nu}(y)}2,
\)
we obtain
\[
\begin{aligned}
&\left|
\sum_y
\bigl(\Lambda(x,y,a,\mu_s)-\Lambda(x,y,a,\nu_s)\bigr)
V_s^{\delta,T,\pmb \nu}(y)
\right|  \\
&\qquad\le
\frac12
\sum_y
|\Lambda(x,y,a,\mu_s)-\Lambda(x,y,a,\nu_s)|
\spn(V_s^{\delta,T,\pmb \nu}) \le
L_\Lambda S_V^\delta\norm{\mu_s-\nu_s}_1.
\end{aligned}
\]
Consequently,
\[
  \sup_{x,a}|Q_s^{\delta,T,\pmb \mu}(x,a)-Q_s^{\delta,T,\pmb \nu}(x,a)|
  \le
  B_\Lambda\norm{V_s^{\delta,T,\pmb \mu}-V_s^{\delta,T,\pmb \nu}}_\infty
  +
  C_\mu^{\delta,T}\norm{\mu_s-\nu_s}_1.
\]
Therefore, by Gr\"onwall's inequality, it holds that
\begin{equation}\label{eq:d}
   \begin{aligned}
      &\norm{ V^{\delta,T,\pmb \mu}_t-V^{\delta,T,\pmb \nu}_t}_\infty\le \int_t^T e^{-\delta(s-t)}\left(B_\Lambda \norm{ V^{\delta,T,\pmb \mu}_s-V^{\delta,T,\pmb \nu}_s}_\infty+C_\mu^{\delta,T} \norm{\mu_s-\nu_s}_1\right)\dd s 
      \\&\implies \norm{ V^{\delta,T,\pmb \mu}_t-V^{\delta,T,\pmb \nu}_t}_\infty\le C_\mu^{\delta,T}\int_t^T e^{-(\delta-B_\Lambda)(s-t)}\norm{\mu_s-\nu_s}_1\dd s.
    \end{aligned} 
\end{equation}

Since \(\norm{\mu_s-\nu_s}_1\le e^{\eta s}d_{\eta,1}^T(\pmb \mu,\pmb \nu)\), multiplying \eqref{eq:d} by \(e^{-\eta t}\) results in
\begin{equation}\label{eq:d2}
  e^{-\eta t}\norm{ V^{\delta,T,\pmb \mu}_t-V^{\delta,T,\pmb \nu}_t}_\infty
  \le C_\mu^{\delta,T} d_{\eta,1}^T(\pmb \mu,\pmb \nu)
  \int_t^T e^{-(\delta-B_\Lambda-\eta)(s-t)}\dd s.
\end{equation}
Thus, using the assumption \(\delta-B_\Lambda-\eta>0\) we obtain
\(
\int_t^T e^{-(\delta-B_\Lambda-\eta)(s-t)}\dd s \le \frac{1}{\delta-B_\Lambda-\eta},
\)
and taking the essential supremum of \eqref{eq:d2} over \(t\in[0,T]\) proves \eqref{eq:V-lip}, as desired.
\end{proof}
The value function estimate yields a Lipschitz bound for \(\Gamma_Q^{\delta,T}\) and \(\Gamma_\Pi^T\). The latter introduces the dependence on \(\alpha\), with larger \(\alpha\) reducing sensitivity to \(Q\)-function perturbations.
\begin{lemma}\label{lem:q-policy}
Under the assumptions of Lemma~\ref{lem:hjb},
\begin{equation}\label{eq:Q-lip}
  d_{\eta,Q}^T(\pmb Q^{\delta,T,\pmb \mu},\pmb Q^{\delta,T,\pmb \nu})
  \le C_\mu^{\delta,T}\left(1+\frac{B_\Lambda}{\delta-B_\Lambda-\eta}\right)d_{\eta,1}^T(\pmb \mu,\pmb \nu).
\end{equation}
Moreover, for any \(Q\)-function flows \(\pmb Q,\pmb Q'\),
\begin{equation}\label{eq:policy-lip}
  d_{\eta,\Pi}^T(\Gamma_\Pi^T(\pmb Q),\Gamma_\Pi^T(\pmb Q'))
  \le
  \frac1{\alpha}d_{\eta,Q}^T(\pmb Q,\pmb Q').
\end{equation}
\end{lemma}

\begin{proof}
The point-wise estimate established in Lemma~\ref{lem:hjb} gives, for a.e. \(t\in[0,T]\),
\begin{equation}\label{eq:ptwise}
  \norm{Q_t^{\delta,T,\pmb \mu}-Q_t^{\delta,T,\pmb \nu}}_\infty
  \le B_\Lambda\norm{V_t^{\delta,T,\pmb \mu}-V_t^{\delta,T,\pmb \nu}}_\infty+C_\mu^{\delta,T}\norm{\mu_t-\nu_t}_1.
\end{equation}
Multiplying by \(e^{-\eta t}\), taking the essential supremum over \(t\), and using Lemma~\ref{lem:hjb} proves \eqref{eq:Q-lip}.

It remains to prove \eqref{eq:policy-lip}. For \(z\in\R^\A\), define
\[
  \sigma_\alpha(z)_i
  =
  \frac{e^{-z_i/\alpha}}{\sum_{j\in\A}e^{-z_j/\alpha}}.
\]
Let \(h=z-z'\) and \(z_\theta=z'+\theta h\), \(\theta\in[0,1]\). A direct
calculation gives
\[
  \frac{d}{d\theta}\sigma_\alpha(z_\theta)_i
  =
  -\frac1\alpha
  \sigma_\alpha(z_\theta)_i
  \left(
      h_i-\sum_{j\in\A}\sigma_\alpha(z_\theta)_j h_j
  \right).
\]
Hence
\[
\begin{aligned}
\norm{\sigma_\alpha(z)-\sigma_\alpha(z')}_1
&\le
\frac1\alpha
\int_0^1
\sum_{i\in\A}
\sigma_\alpha(z_\theta)_i
\left|
      h_i-\sum_{j\in\A}\sigma_\alpha(z_\theta)_j h_j
\right|
\dd\theta .
\end{aligned}
\]
For each fixed \(\theta\), the term inside the sum is the mean absolute
deviation of a random variable taking values \(\{h_i:i\in\A\}\). Since this
random variable is supported on an interval of length at most
\(2\norm{h}_\infty\), its mean absolute deviation is at most
\(\norm{h}_\infty\). Therefore,
\begin{equation}\label{eq:ev-ptwise2}
    \norm{\sigma_\alpha(z)-\sigma_\alpha(z')}_1
  \le
  \frac1\alpha\norm{z-z'}_\infty.
\end{equation}

Applying this point-wise to
\(z=(Q_t(x,a))_{a\in\A}\) and \(z'=(Q'_t(x,a))_{a\in\A}\), then taking the
maximum over \(x\), multiplying by \(e^{-\eta t}\), and taking the essential
supremum over \(t\in[0,T]\), yields
\[d_{\eta,\Pi}^T(\Gamma_\Pi^T(\pmb Q),\Gamma_\Pi^T(\pmb Q'))
  \le
  \frac1{\alpha}d_{\eta,Q}^T(\pmb Q,\pmb Q'),
\]
as desired.
\end{proof}

\subsection{Majorization of \(\Phi_{\delta}^T\) with a positive operator}

The preceding subsection derived Lipschitz estimates for
\(\Gamma_M^T\), \(\Gamma_Q^{\delta,T}\), and \(\Gamma_\Pi^T\). We now use these estimates to analyze the full MFE operator \(\Phi_\delta^T=\Gamma_\Pi^T\circ\Gamma_Q^{\delta,T}\circ\Gamma_M^T\). Instead of collapsing the estimates immediately into a single norm bound, we keep the time-dependent propagation of errors through the forward equation and the discounted HJB equation. This yields a positive integral operator (a nested Volterra operator) that point-wise majorizes the policy error generated by \(\Phi_\delta^T\). Its spectral
radius gives a refined finite-horizon contraction criterion.

\begin{lemma}
\label{lem:finite-horizon-majorization}
Let \(\pmb \pi,\pmb \pi'\in\Pi_T\) be admissible policies, and set
\(
        p_{\pmb \pi,\pmb \pi'}(t):=\|\pi_t-\pi'_t\|_{\Pi}.
\)
Let \(\mathcal A_T^\delta\) be an operator on $(L^{\infty}([0,T]),\|\cdot\|_\infty)$, defined as
\begin{equation}\label{eq:suyosauce}
        (\mathcal A_T^\delta f)(t)
        :=
        c_{\delta,T}
        \left[
            (M_Tf)(t)
            +
            B_\Lambda
            \int_t^T
            e^{-\gamma_\delta(s-t)}
            (M_Tf)(s)\,ds
        \right],
\end{equation}
where
\(\gamma_\delta:=\delta-B_\Lambda,\) \(c_{\delta,T}:=\frac{B_\Lambda C_\mu^{\delta,T}}{\alpha},\)
and \(M_T\) is the multiplication operator
\[
(M_Tf)(t) := \int_0^t e^{a_M(t-s)}f(s)\,ds .
\]
Then, it holds that
\begin{equation}\label{eq:ptwise-ff}
\|\Phi_\delta^T(\pmb \pi)_t-\Phi_\delta^T(\pmb \pi')_t\|_{\Pi}
        \le
        (\mathcal A_T^\delta p_{\pmb \pi,\pmb \pi'})(t),
        \qquad \text{for a.e. } 0\le t\le T .
\end{equation}
\end{lemma}
\begin{proof}
Define $\pmb \mu=\Gamma_M^T(\pmb \pi),$ $\pmb \nu=\Gamma_M^T(\pmb \pi').$
By Lemma \ref{lem:forward}, we must have
\begin{equation}\label{eq:maj-forward}
        \|\mu_t-\nu_t\|_1 
        \le
        B_\Lambda\int_0^t e^{a_M(t-s)}p_{\pmb \pi,\pmb \pi'}(s)\,ds
        =
        B_\Lambda(M_Tp_{\pmb \pi,\pmb \pi'})(t).
\end{equation}
By the same point-wise HJB stability estimate proved in
Lemma~\ref{lem:hjb}, we get
\begin{equation}\label{eq:maj-value}
        \| V^{\delta,T,\pmb \mu}_t-V^{\delta,T,\pmb \nu}_t\|_\infty
        \le
        C_\mu^{\delta,T}
        \int_t^T e^{-(\delta-B_\Lambda)(s-t)}\|\mu_s-\nu_s\|_1\,\dd s .
\end{equation}
For the corresponding \(Q\)-functions to \(\pmb V^{\delta,T,\pmb \mu}\) and \(\pmb V^{\delta,T,\pmb \nu}\), the point-wise estimate \eqref{eq:ptwise} used in
Lemma~\ref{lem:q-policy} gives
\[
        \|Q_t^{\delta,T,\pmb \mu}-Q_t^{\delta,T,\pmb \nu}\|_\infty
        \le
        C_\mu^{\delta,T}\|\mu_t-\nu_t\|_1+B_\Lambda \| V^{\delta,T,\pmb \mu}_t-V^{\delta,T,\pmb \nu}_t\|_\infty.
\]
Using \eqref{eq:maj-value}, we obtain
\begin{equation}\label{eq:maj-Q}
        \|Q_t^{\delta,T,\pmb \mu}-Q_t^{\delta,T,\pmb \nu}\|_\infty
        \le
        C_\mu^{\delta,T}\|\mu_t-\nu_t\|_1
        +
        B_\Lambda C_\mu^{\delta,T}
        \int_t^T e^{-\gamma_\delta(s-t)}\|\mu_s-\nu_s\|_1\,\dd s .
\end{equation}
Substituting the forward bound \eqref{eq:maj-forward} into
\eqref{eq:maj-Q}, we get
\[
\begin{aligned}
        \|Q_t^{\delta,T,\pmb \mu}-Q_t^{\delta,T,\pmb \nu}\|_\infty
        &\le 
        B_\Lambda C_\mu^{\delta,T}
        \left[
            (M_Tp_{\pmb \pi,\pmb \pi'})(t)
            +
            B_\Lambda
            \int_t^T e^{-\gamma_\delta(s-t)}(M_Tp_{\pmb \pi,\pmb \pi'})(s)\,\dd s
        \right].
\end{aligned}
\]
Finally, by the softmin Lipschitz estimate \eqref{eq:ev-ptwise2} in Lemma~\ref{lem:q-policy}, we must have
\[
        \|\Gamma_\Pi^T(\pmb Q^{\delta,T,\pmb \mu})_t
        -
        \Gamma_\Pi^T(\pmb Q^{\delta,T,\pmb \nu})_t\|_{\Pi}
        \le
        \frac{1}{\alpha}
        \|Q_t^{\delta,T,\pmb \mu}-Q_t^{\delta,T,\pmb \nu}\|_\infty .
\]
Since
\( \Phi_\delta^T(\pmb \pi) = \Gamma_\Pi^T(\pmb Q^{\delta,T,\pmb \mu})\), and \( \Phi_\delta^T(\pmb \pi') = \Gamma_\Pi^T(\pmb Q^{\delta,T,\pmb \nu}),
\)
we can conclude that
\begin{align*}
\|\Phi_\delta^T(\pmb \pi)_t-\Phi_\delta^T(\pmb \pi')_t\|_{\Pi}
&\le \frac{B_\Lambda C_\mu^{\delta,T}}{\alpha} \left[ (M_Tp_{\pmb \pi,\pmb \pi'})(t) + B_\Lambda
 \int_t^T e^{-\gamma_\delta(s-t)}(M_Tp_{\pmb \pi,\pmb \pi'})(s)\,\dd s \right]  
 \\&=  (\mathcal A_T^\delta p_{\pmb \pi,\pmb \pi'})(t),
\end{align*}
which proves the desired point-wise majorization \eqref{eq:ptwise-ff}.
\end{proof}

\begin{lemma}\label{lem:kr}
Suppose that \(c_{\delta,T}>0\) and \(B_\Lambda>0\). Then, it holds that
\begin{enumerate}
    \item The operator \(\mathcal A_T^\delta:L^\infty([0,T])\to L^\infty([0,T])\) is a well-defined positive compact operator on
\((L^\infty([0,T]),\|\cdot\|_\infty)\);
    \item Let \(\rho(\mathcal A_T^\delta)\) be the spectral radius of \(\mathcal A_T^\delta\) on \((L^\infty([0,T]),\|\cdot\|_\infty)\). Then, \(\rho(\mathcal A_T^\delta)\) is an eigenvalue of
\(\mathcal A_T^\delta\) with a nonnegative eigenfunction bounded away from \(0\).
This eigenvalue is the \emph{principal eigenvalue}.
\end{enumerate}
\end{lemma}

\begin{proof}
To prove the compactness of the operator \(\mathcal A_T^\delta\), we show that \(\mathcal A_T^\delta\) maps \((L^{\infty}([0,T]),\|\cdot\|_\infty)\) to a relatively compact subset of \((C([0,T]),\|\cdot\|_\infty)\).
For a given \(f\in L^\infty(0,T)\), define the corresponding weighted multiplication operators
\begin{equation}\label{eq:mult-op1}
    m_f(t):=(M_Tf)(t)
    =
    \int_0^t e^{a_M(t-s)}f(s)\,ds
\end{equation}
and
\begin{equation}\label{eq:mult-op2}
    n_f(t):=
    \int_t^T e^{-\gamma_\delta(s-t)}m_f(s)\,ds.
\end{equation}
Then, it holds that \( \mathcal A_T^\delta f = c_{\delta,T}(m_f+B_\Lambda n_f).\)

First, we show that \(\mathcal A_T^\delta\) is a well-defined compact operator
on \((L^\infty([0,T]),\|\cdot\|_\infty)\). Suppose that \(\|f\|_\infty\leq 1\). Since \(T<\infty\),
there is a constant \(C_m>0\), independent of \(f\), such that
\begin{equation}\label{eq:suyo1}
    |m_f(t)|
    \leq
    \int_0^t e^{a_M(t-s)}|f(s)|\,ds
    \leq
    C_m
    \qquad \text{for all } t\in[0,T].
\end{equation}
Moreover, by definition \(m_f\) is absolutely continuous, and hence differentiable a.e. \cite[Lemma 8.2]{brezis2011functional}. In particular, it holds that
\[
    m_f'(t)=a_M m_f(t)+f(t)
    \qquad \text{for a.e. } t\in(0,T).
\]
Hence, by \eqref{eq:suyo1} we obtain
\[
    |m_f'(t)|
    \leq
    a_M C_m+1
    \qquad \text{for a.e. } t\in(0,T);
\]
and thus the family \(\{m_f:\|f\|_\infty\leq 1\}\) is uniformly bounded and
equicontinuous.
Moreover, since \(m_f\) is uniformly bounded, there exists a constant
\(C_n>0\), independent of the choice \(f\), such that
\[
    |n_f(t)|
    \leq
    \int_t^T e^{|\gamma_\delta|(s-t)}|m_f(s)|\,ds
    \leq
    C_n
    \qquad \text{for all } t\in[0,T].
\]
Similarly, one can show that \(\{n_f:\|f\|_\infty\leq 1\}\) is also uniformly bounded and equicontinuous. It follows that for all \( f \in (L^{\infty}([0,T]).\|\cdot\|_{\infty})\), \(\mathcal A_T^\delta f\) is a continuous and bounded function; hence, \(\mathcal A_T^\delta\) is a well-defined operator.

Similarly, it also follows that the set
\( \tilde{\mathcal A}=\{\mathcal A_T^\delta f:\|f\|_\infty\leq 1, f \in L^{\infty}([0,T])\}\)
is uniformly bounded and equicontinuous in \((C([0,T]),\|\cdot\|_{\infty})\). In particular, by the Arzelà--Ascoli theorem, the set \(\tilde{\mathcal A}\) is relatively compact in \((C([0,T]),\|\cdot\|_\infty)\),
and hence also relatively compact in \((L^\infty([0,T]),\|\cdot\|_\infty)\). Therefore, \(\mathcal A_T^\delta\) is a compact operator on \((L^\infty([0,T]),\|\cdot\|_\infty)\) \cite[Section 6]{brezis2011functional}.

Next, we prove the positivity of \(\mathcal A_T^\delta\). Let
\[
K := \{ f \in L^{\infty}([0,T]): f \ge 0 \text{ a.e.}\}.
\]
For any \( f \in L^{\infty}([0,T])\) we can write \(f := f^+ - f^-\) for some \(f^+, f^- \in K\), and thus it follows that \(K\) is a total cone (Definition \ref{def:a1}).
If \(f \in K \), then
\[
m_f(t)
    =
    \int_0^t e^{a_M(t-s)}f(s)\,ds
    \geq 0
\]
for all \(t\in[0,T]\). Similarly, it also follows that \(n_f(t) \geq 0.\)
Furthermore, since \(c_{\delta,T}\geq 0\) and \(B_\Lambda\geq 0\), we obtain
\(\mathcal A_T^\delta f\geq 0\) everywhere on \([0,T]\).
Thus, \(\mathcal A_T^\delta\) is a positive operator (Definition \ref{def:a3}) on \((L^\infty([0,T]),\|\cdot\|_{\infty})\) as we have \(\mathcal A^{\delta}_T(K)\subset K\).

It remains to discuss the spectral radius. By \(\mathbf 1\), we denote the function that is identically \(1\) on \([0,T]\). Assume now that
\(c_{\delta,T}>0\) and \(B_\Lambda>0\). Then, it trivially holds that
\( m_{\mathbf 1}(t) >0\) and \(n_{\mathbf 1}(t) > 0\) for all \( t \in [0,T] \).
Thus, it holds that
\(
(\mathcal A_T^\delta \mathbf 1)(t) = c_{\delta,T} \bigl(m_{\mathbf 1}(t)+B_\Lambda n_{\mathbf 1}(t)\bigr)>0
\)
is strictly positive for every \(t\in[0,T]\). Since
\(\mathcal A_T^\delta \mathbf 1\) is continuous on \([0,T]\), there
exists \(\beta>0\) such that
\(
    \mathcal A_T^\delta \mathbf 1\geq \beta \mathbf 1.
\)
Using positivity of \(\mathcal A_T^\delta\mathbf 1\), we obtain inductively
\( (\mathcal A_T^\delta)^n\mathbf 1
    \geq
    \beta^n \mathbf 1 \) point-wise on \([0,T]\) for every  \(n\geq 1.\)
Let \(\|\cdot\|_{L^\infty\to L^\infty}\) be the operator norm obtained on \((L^{\infty}([0,T]),\| \cdot \|_\infty)\). Notice that
\begin{equation}\label{eq:suyo2}
    \|(\mathcal A_T^\delta)^n\|_{L^\infty\to L^\infty}
    \geq
    \|(\mathcal A_T^\delta)^n\mathbf 1\|_\infty
    \geq
    \beta^n.
\end{equation}
 By Gelfand's formula, \eqref{eq:suyo2} implies that
\(
    \rho(\mathcal A_T^\delta)
    =
    \lim_{n\to\infty}
    \|(\mathcal A_T^\delta)^n\|_{L^\infty\to L^\infty}^{1/n}
    \geq
    \beta
    >
    0.
\)

Finally, since \(\mathcal A_T^\delta\) is a positive compact operator on
the Banach lattice \((L^\infty([0,T]),\|\cdot\|_\infty)\) and
\(\rho(\mathcal A_T^\delta)>0\), the Krein--Rutman theorem (Theorem \ref{thrm:a5}) implies that
\(\rho(\mathcal A_T^\delta)\) is an eigenvalue of
\(\mathcal A_T^\delta\). Moreover, there exists a nonnegative nonzero
function \(p\in L^\infty(0,T)\) such that
\(
    \mathcal A_T^\delta p
    =
    \rho(\mathcal A_T^\delta)p.
\)
 Since
\(\rho(\mathcal A_T^\delta)>0\), the principal eigenvalue can be taken to be strictly positive by the Krein--Rutman theorem (Theorem \ref{thrm:a5}), as desired.
\end{proof}

\subsection{Contraction under Bielecki metrics}

The preceding spectral analysis provides a sharp finite-horizon contraction
criterion for the positive error-propagation operator \(\mathcal A_T^\delta\).
Its principal eigenfunction induces an adapted weighted metric whose
contraction factor is governed by \(\rho(\mathcal A_T^\delta)\), but this metric generally depends on \(T\).
To obtain a metric that can be chosen uniformly over \(T\), we therefore introduce Bielecki metrics with exponential weight \(e^{-\eta t}\). The following result shows that they provide an explicit horizon-independent contraction factor \(\kappa_{\delta,\eta,\alpha}^T\), and optimizing
over \(\eta\) recovers the large-horizon limit of \(\rho(\mathcal A_T^\delta)\).

\begin{theorem}\label{thm:contraction}
Suppose that Assumption~\ref{ass:basic-reg} holds. Let $a_M$ be as in Lemma \ref{lem:F-lip} and $B_{\Lambda}$ as in \eqref{eq:constants}. Assume that there exists $\eta$ such that
\(a_M<\eta<\delta-B_\Lambda.\)
Then, for every \(T<\infty\),
\begin{equation}\label{eq:Phi-lip}
  d_{\eta,\Pi}^T(\Phi_\delta^T(\pmb \pi),\Phi_\delta^T(\pmb \pi'))
  \le \kappa_{\delta,\eta,\alpha}^T d_{\eta,\Pi}^T(\pmb \pi,\pmb \pi'),
\end{equation}
where
\begin{equation}\label{eq:kappa}
  \kappa_{\delta,\eta,\alpha}^T
  =\frac{B_\Lambda}{\eta-a_M}\cdot\frac1{\alpha}\cdot C_\mu^{\delta,T}
  \left(1+\frac{B_\Lambda}{\delta-B_\Lambda-\eta}\right).
\end{equation}
Consequently, if \(\kappa_{\delta,\eta,\alpha}^T<1\), then \(\Phi_\delta^T\) has a unique fixed point in \(\Pi_T\), and the fixed-point iteration converges to it.  The contraction constant \eqref{eq:kappa} is independent of \(T\).
\end{theorem}

\begin{proof}
Let \(\pmb \pi,\pmb \pi'\in\Pi_T\), and set
\(
    \pmb \mu:=\Gamma_M^T(\pmb \pi),\)
\(
    \pmb \nu:=\Gamma_M^T(\pmb \pi').
\)
The inequality \eqref{eq:Phi-lip} follows by tracking how a perturbation of the input policy
propagates through the three components of the MFE operator
\( \Phi_\delta^T=\Gamma_\Pi^T\circ\Gamma_Q^{\delta,T}\circ\Gamma_M^T.\)

The forward equation converts the policy perturbation into a perturbation of
the mean-field flow.  Since \(\eta>a_M\), Lemma~\ref{lem:forward} gives
\[
    d_{\eta,1}^T(\pmb \mu,\pmb \nu)
    \le
    \frac{B_\Lambda}{\eta-a_M}
   d_{\eta,\Pi}^T(\pmb \pi,\pmb \pi').
\]

Next, the span bound in Lemma~\ref{lem:span} gives a uniform control of
\(\pmb V^{\delta,T,\pmb \mu}\), which enters the mean-field sensitivity constant
\(
    C_\mu^{\delta,T}=L_c+L_\Lambda S_V^\delta.
\)
Using this constant, Lemma~\ref{lem:hjb} controls the value-function perturbation, and Lemma~\ref{lem:q-policy} gives the corresponding
\(Q\)-function estimate:
\[
    d_{\eta,Q}^T(\pmb Q^{\delta,T,\pmb \mu},\pmb Q^{\delta,T,\pmb \nu})
    \le
    C_\mu^{\delta,T}
    \left(
        1+\frac{B_\Lambda}{\delta-B_\Lambda-\eta}
    \right)
    d_{\eta,1}^T(\pmb \mu,\pmb \nu).
\]
Here, the condition \(\delta>B_\Lambda\) is absorbed into the stronger requirement \(a_M<\eta<\delta-B_\Lambda\).
Finally, the soft-min map converts the \(Q\)-function perturbation into a
policy perturbation.  By the Lipschitz estimate for \(\Gamma_\Pi^T\),
\begin{align*}
d_{\eta,\Pi}^T(\Phi_\delta^T(\pmb \pi),\Phi_\delta^T(\pmb \pi'))
\le
\frac1{\alpha}
d_{\eta,Q}^T(\pmb Q^{\delta,T,\pmb \mu},\pmb Q^{\delta,T,\pmb \nu})\le
\frac{B_\Lambda}{\eta-a_M}
\cdot
\frac1{\alpha}
\cdot
C_\mu^{\delta,T}
\left(
1+\frac{B_\Lambda}{\delta-B_\Lambda-\eta}
\right)
d_{\eta,\Pi}^T(\pmb \pi,\pmb \pi') .
\end{align*}
This proves \eqref{eq:Phi-lip} with contraction constant
\(\kappa_{\delta,\eta,\alpha}^T\).
\end{proof}

\begin{remark}\label{rem:eta}
The forward Kolmogorov estimate produces the weighted kernel
\(e^{-(\eta-a_M)(t-s)}, 0\le s\le t,\)
and the backward HJB estimate produces
\( e^{-(\delta-B_\Lambda-\eta)(s-t)}, t\le s\le T.\)
Thus uniform-in-\(T\) control requires
\(a_M<\eta<\delta-B_\Lambda.\) If this interval is empty, this Bielecki-metric estimate does not yield a horizon-independent contraction condition.
\end{remark}

\begin{lemma} \label{lem:contraction}
Suppose that \(\rho(\mathcal A_T^\delta)<1\), and let
\(r_T\in L^\infty([0,T])\) be a strictly positive eigenfunction, bounded away from \(0\) and satisfying
\(
    \mathcal A_T^\delta r_T
    =
    \rho(\mathcal A_T^\delta)r_T,
\)
which exists by Lemma \ref{lem:kr}. Define
\[
    d_{r_T,\Pi}(\pmb \pi,\pmb \pi')
    :=
    \esssup_{0\le t\le T}
    \frac{\norm{\pi_t-\pi'_t}_{\Pi}}{r_T(t)}.
\]
Then, \(\Phi_\delta^T\) is a contraction on
\((\Pi_T,d_{r_T,\Pi})\) with contraction constant
\(\rho(\mathcal A_T^\delta)\).
\end{lemma}

\begin{proof}
Since \(r_T\) is strictly positive and bounded away from zero, the metric
\(d_{r_T,\Pi}\) is equivalent to the usual \(L^\infty\)-metric on \(\Pi_T\).
Hence \((\Pi_T,d_{r_T,\Pi})\) is complete.

Let \(\pmb \pi,\pmb \pi'\in\Pi_T\) and set
\[
    p_{\pmb \pi,\pmb \pi'}(t):=\norm{\pi_t-\pi'_t}_{\Pi}.
\]
By the definition of \(d_{r_T,\Pi}\), we have
\(
    p_{\pmb \pi,\pmb \pi'}(t)
    \le
    d_{r_T,\Pi}(\pmb \pi,\pmb \pi')\,r_T(t)\)
for a.e. \(t\in[0,T].
\)
Using the point-wise majorization from
Lemma~\ref{lem:finite-horizon-majorization} and positivity of
\(\mathcal A_T^\delta\), we obtain
\[
\norm{\Phi_\delta^T(\pmb \pi)_t-\Phi_\delta^T(\pmb \pi')_t}_{\Pi}
\le
(\mathcal A_T^\delta p_{\pmb \pi,\pmb \pi'})(t)\le
d_{r_T,\Pi}(\pmb \pi,\pmb \pi')
(\mathcal A_T^\delta r_T)(t)=
\rho(\mathcal A_T^\delta)
d_{r_T,\Pi}(\pmb \pi,\pmb \pi')r_T(t).
\]
Dividing by \(r_T(t)\) and taking the essential supremum over
\(t\in[0,T]\), we get
\[
    d_{r_T,\Pi}(\Phi_\delta^T(\pmb \pi),\Phi_\delta^T(\pmb \pi'))
    \le
    \rho(\mathcal A_T^\delta)
    d_{r_T,\Pi}(\pmb \pi,\pmb \pi').
\]
Since \(\rho(\mathcal A_T^\delta)<1\), this proves that \(\Phi_\delta\) is a
contraction on \((\Pi_T,d_{r_T,\Pi})\).
\end{proof}

The next theorem makes this comparison precise. The large-horizon limit of \(\rho(\mathcal A_T^\delta)\) equals the infimum of the contraction factors over the admissible Bielecki weights.

\begin{theorem}\label{thrm:sauce2}
Let Assumption~\ref{ass:basic-reg} hold and $B_{\Lambda}$ be as in \eqref{eq:constants}. 
Suppose that \(\delta-B_\Lambda>a_M.\)
Let \(\mathcal A_T^\delta\) be the point-wise majorization operator defined on \((L^\infty([0,T]),\|\cdot\|_\infty)\) in \eqref{eq:suyosauce}.
By \(\rho(\mathcal A_T^\delta)\) denote the spectral radius of
\(\mathcal A_T^\delta\) on \((L^\infty([0,T]),\|\cdot\|_\infty)\).
Then, it holds that
\[
\lim_{T\to\infty}\rho(\mathcal A_T^\delta) = \inf_{a_M<\eta<\delta-B_\Lambda} \kappa_{\delta,\eta,\alpha}^T = c_{\delta,\infty}\left(\sqrt{\delta-a_M}-\sqrt{B_\Lambda}\right)^{-2}.
\]
\end{theorem}

\begin{proof}
For readability, let
\(a:=a_M,\) \(B:=B_\Lambda,\) \(\gamma:=\delta-B_\Lambda.\)
By the running assumption, we have \(B>0\) and \(\gamma>a\).  Let \(\lambda_T>0\) be the eigenvalue
of \(\mathcal A_T^\delta\) with a nonnegative nonzero eigenfunction \(p\), that is
\( \mathcal A_T^\delta p = \lambda_T p,\) which exists by Lemma \ref{lem:kr}. Notice that \(\lim_{T \to \infty} \lambda_T=\lambda_\infty\) exists and finite.
To save notation, introduce $m :=  m_p$ and $n := n_p$, where \(m_p\) and \(n_p\) are as defined in \eqref{eq:mult-op1} and \eqref{eq:mult-op2}, respectively.
Then
\[
(\mathcal A_T^\delta p)(t) = c_{\delta,T}\bigl(m(t)+Bn(t)\bigr).
\]

Set \( \theta_T:=\frac{c_{\delta,T}}{\lambda_T}.\)
The eigenvalue equation associated with the operator \(\mathcal A_T^\delta\) is therefore
\begin{equation}\label{eq:eig}
        p(t)=\theta_T\bigl(m(t)+Bn(t)\bigr).
\end{equation}
Moreover,
\(
m'(t)=am(t)+p(t),\) \(m(0)=0, \) and 
\(
n'(t)=-m(t)+\gamma n(t),\) \(n(T)=0.
\)
Substituting \(p \equiv \theta_T(m+Bn)\), we obtain the following linear time-invariant dynamical system
\begin{equation}\label{eq:ODE}
        \frac{d}{dt}
        \begin{pmatrix}
            m(t)\\ n(t)
        \end{pmatrix}
        =
        M_{\theta_T}
        \begin{pmatrix}
            m(t)\\ n(t)
        \end{pmatrix},
        \qquad
        M_{\theta_T}
        :=
        \begin{pmatrix}
            a+\theta_T & \theta_T B\\
            -1 & \gamma
        \end{pmatrix},
\end{equation}
with the following Dirichlet boundary conditions (BCs):
\begin{equation}\label{eq:ev-end}
        m(0)=0,
        \qquad
        n(T)=0.
\end{equation}
Note that under the BCs \eqref{eq:ev-end}, the equation \eqref{eq:ODE} has the following closed-form solution:
\begin{equation}\label{eq:ev-eq}
\begin{pmatrix}
    m(T)\\
    n(T)
\end{pmatrix} = e^{M_{\theta_T}T} \begin{pmatrix}
    m(0)\\
    n(0)
\end{pmatrix} = n(0)\begin{pmatrix}
    \bigl(e^{M_{\theta_T} T}\bigr)_{12}\\
    \bigl(e^{M_{\theta_T} T}\bigr)_{22}
\end{pmatrix};
\end{equation}
see \cite[Section 3.2]{basar2020lecture}.
Thus, the BC $n(T) = 0 = n(0)\bigl(e^{M_{\theta_T} T}\bigr)_{22}$ implies that a nontrivial solution to \eqref{eq:ev-eq} exists if, and only if,
\begin{equation}\label{eq:ev-bc}
        \bigl(e^{M_{\theta_T} T}\bigr)_{22}=0.
\end{equation}

Let
\(
        d(\theta):=a+\theta-\gamma,
\)
and
\(
        \Delta(\theta):=
        \sqrt{d(\theta)^2-4B\theta}.
\)
Set \(\tau:=\operatorname{tr}(M_\theta)=a+\theta+\gamma\) and
\(N_\theta:=M_\theta-\frac{\tau}{2}I\).  Then, with this notation, notice that 
\begin{equation}\label{eq:ev-f}
e^{M_\theta T} = e^{\frac{1}{2}(a+\theta+\gamma)T}e^{N_\theta T},
\end{equation}
and we can write
\begin{equation}\label{eq:n}
N_\theta =
\begin{pmatrix}
\frac{d(\theta)}{2} & \theta B\\
-1 & -\frac{d(\theta)}{2}
\end{pmatrix}.
\end{equation}

Let $I$ be the $2\times 2$ dimensional identity matrix. By \(N^m_\theta\), we denote the \(m^{\mathrm{th}}\) power of the matrix \(N_\theta\). Note that \eqref{eq:n} implies
\[
N_\theta^2
=
\frac{d(\theta)^2-4B\theta}{4}I
=
\frac{\Delta(\theta)^2}{4}I \implies 
N_\theta^{2k}
=
\left(\frac{\Delta(\theta)^2}{4}\right)^k I, \text{ and }
N_\theta^{2k+1}
=
\left(\frac{\Delta(\theta)^2}{4}\right)^k N_\theta.
\]
In particular, using the definitions of hyperbolic $\mathrm{sine}$ and $\mathrm{cosine}$, it holds that
\begin{align*}
e^{N_\theta T} &= \sum_{j=0}^{\infty} \frac{(N_\theta T)^j}{j!} = \sum_{k=0}^{\infty} \frac{N_\theta^{2k}T^{2k}}{(2k)!} + \sum_{k=0}^{\infty} \frac{N_\theta^{2k+1}T^{2k+1}}{(2k+1)!}
\\& =
\cosh\left(\frac{\Delta(\theta)T}{2}\right)I +
\frac{2}{\Delta(\theta)}
\sinh\left(\frac{\Delta(\theta)T}{2}\right)N_\theta 
\end{align*}
Consequently, \eqref{eq:ev-f} can be represented as
\[
e^{M_\theta T}
=
e^{\frac{\tau T}{2}}
\left[
\cosh\left(\frac{\Delta(\theta)T}{2}\right)I
+
\frac{2}{\Delta(\theta)}
\sinh\left(\frac{\Delta(\theta)T}{2}\right)N_\theta
\right].
\]
Taking the \((2,2)\)-entry of $e^{M_{\theta} T}$ and using
\( (N_\theta)_{22}=-\frac{d(\theta)}{2}\), we obtain
\[
\bigl(e^{M_\theta T}\bigr)_{22}
=
e^{\frac12(a+\theta+\gamma)T}
\left[
\cosh\left(\frac{\Delta(\theta)T}{2}\right)
-
\frac{d(\theta)}{\Delta(\theta)}
\sinh\left(\frac{\Delta(\theta)T}{2}\right)
\right].
\]
Thus, the nondegeneracy condition \eqref{eq:ev-bc} for this eigenvalue problem is equivalent to the following hyperbolic root problem
\begin{equation}\label{eq:1-ev}
        \cosh\left(\frac{\Delta(\theta)T}{2}\right)
        -
        \frac{a+\theta-\gamma}{\Delta(\theta)}
        \sinh\left(\frac{\Delta(\theta)T}{2}\right)
        =
        0 .
\end{equation}
The spectral radius corresponds to the eigenvalue obtained under the smallest positive root
\(\theta_T\) of \eqref{eq:1-ev}, and hence
\begin{equation}\label{eq:2-ev}
\rho(\mathcal A_T^\delta) = c_{\delta,T}\theta_T^{-1}.
\end{equation}

As a consequence of \eqref{eq:2-ev}, to obtain the asymptote of \(\rho(\mathcal A_T^\delta)\), it suffices to study the asymptotes of the sequence \((\theta_T)_{T \in \mathbb N}\).  Let
\( \beta:=\gamma-a>0.\)
Then, it follows that
\( \Delta(\theta)^2 = (a+\theta-\gamma)^2-4B\theta = (\theta-\beta)^2-4B\theta\), or equivalently,
\( \Delta(\theta)^2 = (\theta-\theta_-)(\theta-\theta_+),\)
where
\[
        \theta_\pm
        =
        \beta+2B\pm 2\sqrt{B(\beta+B)}.
\]
In particular,
\(
        \theta_-
        =
        \beta+2B-2\sqrt{B(\beta+B)}
        =
        \bigl(\sqrt{\beta+B}-\sqrt B\bigr)^2 .
\)
Since
\(
\beta+B = \gamma-a+B = \delta-a,
\)
we have
\begin{equation}\label{eq:3-ev}
\theta_-  = \bigl(\sqrt{\delta-a}-\sqrt B\bigr)^2 .
\end{equation}

If \(0<\theta<\theta_-\), then \(\Delta(\theta)>0\) and
\(
a+\theta-\gamma = \theta-\beta < 0.
\)
Therefore
\[
\cosh\left(\frac{\Delta(\theta)T}{2}\right) -
\frac{\theta-\beta}{\Delta(\theta)}
 \sinh\left(\frac{\Delta(\theta)T}{2}\right)
 > 0,
\]
and thus \eqref{eq:1-ev} has no positive root below \(\theta_-\).  It follows now that \( \theta_T\geq \theta_- .\)

For \(\theta\in(\theta_-,\theta_+)\), write
\(\Delta(\theta)=i\omega(\theta),\) and \(\omega(\theta) :=\sqrt{-(\theta-\theta_-)(\theta-\theta_+)}.\)
Note that, \eqref{eq:1-ev} is equivalent to the trigonometric equation
\begin{equation}\label{eq:5-ev}
\cos\left(\frac{\omega(\theta)T}{2}\right)
 -
 \frac{\theta-\beta}{\omega(\theta)}\sin\left(\frac{\omega(\theta)T}{2}\right)=0 \iff \theta-\beta
        =
        \omega(\theta)
        \cot\left(\frac{\omega(\theta)T}{2}\right)
\end{equation}
when \(\theta\in(\theta_-,\theta_+)\).
The smallest positive root that satisfies either equality in \eqref{eq:5-ev} above for a given \(T\), \(\theta_T\), lies on the first negative
branch of the cotangent. Thus, for all \(T \in \mathbb N\) we must have
\[
\frac{\omega(\theta_T)T}{2} \in \left(\frac{\pi}{2},\pi\right) \iff \omega(\theta_T)\leq 2\pi T^{-1},
\]
which also implies that \(\cot\left(\frac{\omega(\theta_T)T}{2}\right) < 0\). In particular, for all \(T \in \mathbb N\) we have \(\theta_T < \beta\). Note that, since \(\theta_- < \beta < \theta_+\), we cannot have \(\lim_{T \to \infty} \theta_T = \theta_+\).
Using the definition of \(\omega(\theta_T)\), we now obtain
\begin{equation}\label{eq:6-ev}
\omega(\theta)^2 = (\theta-\theta_-)(\theta_+-\theta) \implies  0\le (\theta_T-\theta_-)(\theta_+-\theta_T) \leq 4\pi^2T^{-2}.
\end{equation}
Thus, the following limit holds
\begin{equation}\label{eq:7-ev}
        \lim_{T \to \infty} \theta_T = \theta_-.
\end{equation}
Combining \eqref{eq:2-ev}, \eqref{eq:3-ev}, and \eqref{eq:7-ev}, and using
\(c_T\to c_{\delta,\infty}\), yields
\begin{equation}\label{eq:8-ev}
\lim_{T\to\infty} \rho(\mathcal A_T^\delta) = c_{\delta,\infty}\theta_-^{-1} =
        c_{\delta,\infty}\left(\sqrt{\delta-a}-\sqrt B\right)^{-2} = (\ast).
\end{equation}

It remains to identify $(\ast)$ as a contraction condition obtained from the Bielecki metric.  Let
\( x:=\eta-a,\) and \( \beta:=\gamma-a.\)
Then, \(a<\eta<\gamma\) if, and only if, \(0<x<\beta .\) With this notation, we must have \(\kappa_{\delta,\eta,\alpha}^T = c_{\delta,\infty} f(x)\), where
\[
f(x):= \frac{1}{x} \left(1+\frac{B}{\beta-x}\right).
\]
Now, allowing $x$ to vary, a direct computation gives
\[
        f'(x)
        =
        -\frac{1}{x^2}
        -
        \frac{B(\beta-2x)}{x^2(\beta-x)^2}.
\]
Thus, the unique critical point of \(f\) in the interval \((0,\beta)\) is
\( x_* = \beta+B-\sqrt{B(\beta+B)},\)
which leads to
\(
\inf_{0<x<\beta}f(x) = \left(\sqrt{\beta+B}-\sqrt B\right)^{-2}.
\)
Since \(\beta+B=\delta-a\), by \eqref{eq:8-ev}, we obtain
\[
\inf_{a<\eta<\gamma} \kappa_{\delta,\eta,\alpha}^T =
 c_{\delta,\infty}\left(\sqrt{\delta-a}-\sqrt B\right)^{-2} \implies         \lim_{T\to\infty}\rho(\mathcal A_T^\delta)
        =
        \inf_{a_M<\eta<\delta-B_\Lambda}
        \kappa_{\delta,\eta,\alpha}^T,
\]
as desired.
\end{proof}

The proof of Theorem \ref{thrm:sauce2} provides a quantitative method to calculate a \(T\)-dependent contraction condition. To emphasize this, we rewrite this consequence as a separate result.

\begin{corollary}\label{cor:3.13}
    Let \(\Delta(\theta)^2=(a_M+\theta-\gamma_\delta)^2-4B_\Lambda\theta \). By \(\theta_T\), denote the smallest positive root of
    \[
\cos\!\left(\frac{\Delta(\theta)T}{2}\right)
-
\frac{a_M+\theta-\gamma_\delta}{\Delta(\theta)}
\sin\!\left(\frac{\Delta(\theta)T}{2}\right)=0.
\]
Then, \(\rho(A^{\delta}_T)=c_T\theta_T^{-1}.\)
\end{corollary}
\begin{proof}
    Note that existence of a positive root of this equation follows from Lemma \ref{lem:kr}. This result directly follows from the proof of Theorem \ref{thrm:sauce2}. 
\end{proof}

We now prove the main theorem in the beginning of this section by combining the contraction analysis and then conclude the spectral comparison above. 

\begin{proof}[Proof of Theorem~\ref{thrm:main}]
Fix \(T<\infty\). By the well-posedness of
\(\Gamma_M^T\), \(\Gamma_Q^{\delta,T}\), and \(\Gamma_\Pi^T\), the composition \( \Phi_\delta^T
    =
    \Gamma_\Pi^T\circ\Gamma_Q^{\delta,T}\circ\Gamma_M^T\)
is a self-map on \(\Pi_T\). By Theorem~\ref{thm:contraction}, for all
\(\pmb \pi,\pmb \pi'\in\Pi_T\),
\(
  d_{\eta,\Pi}^T(\Phi_\delta^T(\pmb \pi),\Phi_\delta^T(\pmb \pi'))
  \le
  \kappa_{\delta,\eta,\alpha}^T
 d_{\eta,\Pi}^T(\pmb \pi,\pmb \pi').
\)
Thus, if \(\kappa_{\delta,\eta,\alpha}^T<1\), \(\Phi_\delta^T\) is a
contraction on the complete metric space \((\Pi_T,d_{\eta,\Pi}^T)\). Banach's
fixed-point theorem leads to a unique fixed point
\(\pmb \pi^{\delta,T}\in\Pi_T\), and the iteration
\(
    \pmb \pi^{k+1}=\Phi_\delta^T(\pmb \pi^k)
\)
converges to \(\pmb \pi^{\delta,T}\) starting from \(\pmb \pi^0\in\Pi_T\). Setting
\(
    \pmb \mu^{\delta,T}:=\Gamma_M^T(\pmb \pi^{\delta,T}), 
\)
 \((\pmb \pi^{\delta,T},\pmb \mu^{\delta,T})\) is the unique finite-horizon \((\delta,\alpha)\)-regularized MFE to \(\mathrm{MFG}_{\delta,T,\alpha}\).
\end{proof}

Thus, the spectral metric is adapted to each fixed horizon, whereas
Bielecki metrics use an exponential weight that is independent of \(T\).
Under the assumptions above, the best contraction factor attainable within the Bielecki family coincides with the large-horizon limit of
\(\rho(\mathcal A_T^\delta)\). Hence, the horizon-independent Bielecki construction is asymptotically sharp relative to the finite-horizon spectral criterion.
\subsection{Numerical verification of the spectral-radius asymptote}

We illustrate Theorem~\ref{thrm:sauce2} with the numerical values 
\[
(a_M,B_\Lambda,\delta,c_\delta)=(0.2,0.3,1.5,0.25),
\qquad
c_\delta:=B_\Lambda C_\mu^{\delta,T}/\alpha .
\]
Then \(\gamma_\delta:=\delta-B_\Lambda=1.2>a_M\), and thus the admissible range
\(a_M<\eta<\delta-B_\Lambda\) is nonempty. The optimized Bielecki rate is
\[
\inf_{a_M<\eta<\delta-B_\Lambda}\kappa_{\delta,\eta,\alpha}^T
=
c_\delta\bigl(\sqrt{\delta-a_M}-\sqrt{B_\Lambda}\bigr)^{-2}
\approx 0.71225 .
\]

With the notation used in the proof of Theorem \ref{thrm:sauce2}, for each \(T\), we compute, and display in Table \ref{tab:spectral-bielecki}, \(\rho(\mathcal A_T^\delta)=c_\delta/\theta_T\),
where \(\theta_T\) is the smallest positive root of
\[
\cos\!\left(\frac{\Delta(\theta)T}{2}\right)
-
\frac{a_M+\theta-\gamma_\delta}{\Delta(\theta)}
\sin\!\left(\frac{\Delta(\theta)T}{2}\right)=0,
\qquad
\Delta(\theta)^2=(a_M+\theta-\gamma_\delta)^2-4B_\Lambda\theta .
\]

\begin{table}[h]
\centering
\small
\begin{tabular}{c|ccccc}
\hline
\(T\) & 5 & 10 & 20 & 50 & 100 \\
\hline
\(\rho(\mathcal A_T^\delta)\) & 0.39107 & 0.55720 & 0.65595 & 0.70102 & 0.70924 \\
\hline
\end{tabular}
\caption{Finite-horizon spectral rates approaching the optimized Bielecki rate \(0.71225\).}
\label{tab:spectral-bielecki}
\end{table}

\section{Infinite-Horizon Contraction and Finite-to-Infinite Error Bounds}
\label{sect:infinite-error}
In this section, we study the contraction properties of infinite-horizon non-stationary \(\delta\)-discounted \(\alpha\)-regularized  MFGs and their relation to the finite-horizon problem. In the discrete-time setting, the majorizing dynamics are represented by matrices, and constructing an appropriate weighted space that captures the desired spectral radius in the infinite-horizon case requires additional restrictions compared to the finite-horizon case \cite{aydin2025approximation}. This gap arises because the majorizing dynamics are no longer compact operators in the infinite-horizon setting. We show that this phenomenon does not occur in continuous time and that the spectral radius that characterizes the contraction in the infinite-horizon case coincides with that of the asymptotics of the spectral radii obtained in the finite-horizon case. We also provide a rate of convergence result between finite- and infinite-horizon \((\delta,\alpha)\)-regularized MFE.

\subsection{Contraction of the infinite-horizon non-stationary MFE operator}
We first derive a contraction condition for infinite-horizon non-stationary model \(\mathrm{MFG}_{\delta,\infty,\alpha}\) introduced in Section~\ref{subsec:infinite-ns}. The
argument largely follows the finite-horizon analysis, since Assumption~\ref{ass:basic-reg} provides pointwise-in-time estimates for the
forward mean-field function, the Hamiltonian, the \(Q\)-function, and the soft-min map, with constants independent of the horizon. The main difference is that the relevant intervals are now taken over \([0,\infty)\), where discounting ensures their finiteness.

The main result of this subsection is the following contraction condition:
\begin{theorem}\label{thm:infty-contraction}
Suppose that Assumption~\ref{ass:basic-reg} holds. Let $a_M$ be as in Lemma \ref{lem:F-lip} and $B_{\Lambda}$ as in \eqref{eq:constants}. Assume that there exists $\eta$ such that
\(
    a_M<\eta<\delta-B_\Lambda
\)
and
\[
  \kappa_{\delta,\eta,\alpha}^{\infty}
  :=
  \frac{B_\Lambda}{\eta-a_M}
  \cdot
  \frac1{\alpha}
  \cdot
  \left(
      L_c
      +
      L_\Lambda
      \frac{2M_c}{\delta-B_\Lambda}
  \right)
  \left(
      1+\frac{B_\Lambda}{\delta-B_\Lambda-\eta}
  \right)
  <1.
\]
Then, the infinite-horizon MFE operator
\(
    \Phi_\delta^\infty:
    (\Pi_\infty,d_{\eta,\Pi}^{\infty})
    \to
    (\Pi_\infty,d_{\eta,\Pi}^{\infty})
\)
corresponding to \(\mathrm{MFG}_{\delta,\infty,\alpha}\)
is a contraction with contraction factor 
\(\kappa_{\delta,\eta,\alpha}^{\infty}\).  Consequently, there exists a unique 
infinite-horizon non-stationary \((\delta,\alpha)\)-regularized MFE \((\pmb \pi^{\delta,\infty},\pmb \mu^{\delta,\infty})\).  Moreover, for every \(\pmb \pi^0\in\Pi_\infty\), the iteration
\(
    \pmb \pi^{k+1}=\Phi_\delta^\infty(\pmb \pi^k)
\)
converges to the unique infinite-horizon non-stationary \((\delta,\alpha)\)-regularized MFE under
\(d_{\eta,\Pi}^{\infty}\), with
\(\pmb \mu^{\delta,\infty}=\Gamma_M^{\infty}(\pmb \pi^{\delta,\infty})\).
\end{theorem}

The proof follows the same strategy as Theorem~\ref{thrm:main}, with the necessary modifications for integration over \([0,\infty)\). We record the corresponding forward function, HJB, \(Q\)-function, and policy estimates before combining them.

\begin{lemma}\label{lem:forward-infty}
Let \(\pmb \mu=\Gamma_M^\infty(\pmb \pi)\) and
\(\pmb \nu=\Gamma_M^\infty(\pmb \pi')\).  If \(\eta>a_M\), then
\begin{equation}\label{eq:forward-infty-lip}
  d_{\eta,1}^{\infty}(\pmb \mu,\pmb \nu)
  \le
  B_\Lambda(\eta-a_M)^{-1}
  d_{\eta,\Pi}^{\infty}(\pmb \pi,\pmb \pi').
\end{equation}
\end{lemma}

\begin{proof}
For every \(t\ge0\), the restrictions of \(\pmb \mu\) and \(\pmb \nu\) to \([0,t]\)
solve the corresponding finite-horizon forward equations with the same initial
condition \(\mu_0\).  Hence, by the same argument as in Lemma~\ref{lem:forward},
\[
  \norm{\mu_t-\nu_t}_1
  \le
  B_\Lambda\int_0^t e^{a_M(t-s)}
  \norm{\pi_s-\pi'_s}_{\Pi}\,\dd s \text{ and } \norm{\pi_s-\pi'_s}_{\Pi}
  \le
  e^{\eta s}d_{\eta,\Pi}^{\infty}(\pmb \pi,\pmb \pi')
  \quad\text{for a.e. }s\ge0.
\]
Thus, we obtain
\[
\begin{aligned}
  e^{-\eta t}\norm{\mu_t-\nu_t}_1
  &\le
  B_\Lambda d_{\eta,\Pi}^{\infty}(\pmb \pi,\pmb \pi')
  \int_0^t e^{-(\eta-a_M)(t-s)}\,\dd s \le
  \frac{B_\Lambda}{\eta-a_M}
  d_{\eta,\Pi}^{\infty}(\pmb \pi,\pmb \pi'),
\end{aligned}
\]
where the last inequality uses \(\eta>a_M\).  Taking the essential supremum
over \(t\ge0\) proves \eqref{eq:forward-infty-lip}.
\end{proof}

We next establish the infinite-horizon counterpart of the finite-horizon
span estimate. Since there is no terminal cost, the resulting bound contains
no contribution from \(\spn(g)\).

\begin{lemma}\label{lem:span-infty}
Suppose \(\delta>B_\Lambda\).  Then, for every
\(\pmb \mu\in\mathcal M_\infty(\mu_0)\) and every \(t\ge0\),
\begin{equation}\label{eq:span-infty}
  \spn(V_t^{\delta,\infty,\pmb \mu})
  \le
  S_V^{\delta,\infty}
  :=
  2M_c(\delta-B_\Lambda)^{-1}.
\end{equation}
\end{lemma}

\begin{proof}
Using the representation \eqref{eq:HJB-infty-mild} and the estimate \( |H_t^{\pmb \mu}(x,v)-H_t^{\pmb \mu}(z,v)|
  \le
  2M_c+B_\Lambda\spn(v)\), proved as in Lemma~\ref{lem:span}, we obtain
\[
\begin{aligned}
  \spn(V_t^{\delta,\infty,\pmb \mu})
  &\le
  \int_t^\infty e^{-\delta(s-t)}
  \left(
      2M_c
      +
      B_\Lambda\spn(V_s^{\delta,\infty,\pmb \mu})
  \right)\dd s .
\end{aligned}
\]
The solution \(\pmb V^{\delta,\infty,\pmb \mu}\) is bounded, hence
\(
  R:=\esssup_{s\ge0}\spn(V_s^{\delta,\infty,\pmb \mu})<\infty .
\)
Taking the essential supremum in the previous inequality leads to
\(
  R
  \le
  2M_c \delta^{-1}
  +
  B_\Lambda\delta^{-1}R.
\)
Since \(\delta>B_\Lambda\), rearranging yields
\(
  R\le 2M_c(\delta-B_\Lambda)^{-1}.
\)
This proves \eqref{eq:span-infty}.
\end{proof}

Define
\begin{equation}\label{eq:Cmu-infty}
  C_\mu^{\delta,\infty}
  :=
  L_c+L_\Lambda S_V^{\delta,\infty}
  =
  L_c+L_\Lambda 2M_c(\delta-B_\Lambda)^{-1}.
\end{equation}

Next, we establish the Lipschitz properties of the infinite-horizon discounted HJB equation with respect to the mean-field flow.

\begin{lemma}\label{lem:hjb-infty}
Let \(\pmb \mu,\pmb \nu\in\mathcal M_\infty(\mu_0)\). If \(\eta<\delta-B_\Lambda\), then
\begin{equation}\label{eq:V-infty-lip}
  d_{\eta,V}^{\infty}
  (\pmb V^{\delta,\infty,\pmb \mu},\pmb V^{\delta,\infty,\pmb \nu})
  \le
  C_\mu^{\delta,\infty}
       (\delta-B_\Lambda-\eta)^{-1}
  d_{\eta,1}^{\infty}(\pmb \mu,\pmb \nu).
\end{equation}
\end{lemma}

\begin{proof}
Let
\(
  d_t^\infty:=
  \norm{
  V_t^{\delta,\infty,\pmb \mu}
  -
  V_t^{\delta,\infty,\pmb \nu}
  }_\infty,
\)
and
\(
  m_t^\infty:=\norm{\mu_t-\nu_t}_1 .
\)
Applying the same Hamiltonian decomposition as in
Lemma~\ref{lem:hjb} to the integral representation
\eqref{eq:HJB-infty-mild} gives
\[
  d_t^\infty
  \le
  \int_t^\infty e^{-\delta(s-t)}
  \left(
      B_\Lambda d_s^\infty
      +
      C_\mu^{\delta,\infty}m_s^\infty
  \right)\dd s .
\]
Multiplying by \(e^{-\eta t}\) and noting that
\(
  d_s^\infty\le e^{\eta s}
  d_{\eta,V}^{\infty}
  (\pmb V^{\delta,\infty,\pmb \mu},\pmb V^{\delta,\infty,\pmb \nu}),
\)
and
\(
  m_s^\infty\le e^{\eta s}d_{\eta,1}^{\infty}(\pmb \mu,\pmb \nu),
\)
\begin{align}\label{eq:quaso:discounted}
  e^{-\eta t}d_t^\infty
  &\le
  \left(
      B_\Lambda
      d_{\eta,V}^{\infty}
      (\pmb V^{\delta,\infty,\pmb \mu},\pmb V^{\delta,\infty,\pmb \nu})
      +
      C_\mu^{ \delta,\infty}
      d_{\eta,1}^{\infty}(\pmb \mu,\pmb \nu)
  \right)
  \int_0^\infty e^{-(\delta-\eta)r}\,\dd r .
\end{align}
Taking the essential supremum of \eqref{eq:quaso:discounted} over \(t\ge0\), we get
\[
  d_{\eta,V}^{\infty}
  (\pmb V^{\delta,\infty,\pmb \mu},\pmb V^{\delta,\infty,\pmb \nu})
  \le
  B_\Lambda(\delta-\eta)^{-1}
  d_{\eta,V}^{\infty}
  (\pmb V^{\delta,\infty,\pmb \mu},\pmb V^{\delta,\infty,\pmb \nu})
  +
  C_\mu^{\delta,\infty}(\delta-\eta)^{-1}
  d_{\eta,1}^{\infty}(\pmb \mu,\pmb \nu).
\]
Since \(\eta<\delta-B_\Lambda\), we have
\(B_\Lambda/(\delta-\eta)<1\). Thus, a straightforward rearrangement proves
\eqref{eq:V-infty-lip}.
\end{proof}

The corresponding \(Q\)-function and policy estimates follow immediately.

\begin{lemma}\label{lem:q-policy-infty}
Under the assumptions of Lemma~\ref{lem:hjb-infty},
\begin{equation}\label{eq:Q-infty-lip}
  d_{\eta,Q}^{\infty}
  (\pmb Q^{\delta,\infty,\pmb \mu},\pmb Q^{\delta,\infty,\pmb \nu})
  \le
  C_\mu^{\delta,\infty}
  \left(
      1+B_\Lambda(\delta-B_\Lambda-\eta)^{-1}
  \right)
  d_{\eta,1}^{\infty}(\pmb \mu,\pmb \nu).
\end{equation}
Moreover, for any \(\pmb Q,\pmb Q'\in\mathcal Q_\infty^\eta\),
\begin{equation}\label{eq:policy-infty-lip}
  d_{\eta,\Pi}^{\infty}
  (\Gamma_\Pi^\infty(\pmb Q),\Gamma_\Pi^\infty(\pmb Q'))
  \le
  \alpha^{-1}
  d_{\eta,Q}^{\infty}(\pmb Q,\pmb Q').
\end{equation}
\end{lemma}

\begin{proof}
For a.e. \(t\ge0\), the definition of the \(Q\)-function gives
\[
  \norm{
  Q_t^{\delta,\infty,\pmb \mu}
  -
  Q_t^{\delta,\infty,\pmb \nu}
  }_\infty
  \le
  B_\Lambda
  \norm{
  V_t^{\delta,\infty,\pmb \mu}
  -
  V_t^{\delta,\infty,\pmb \nu}
  }_\infty
  +
  C_\mu^{\delta,\infty}\norm{\mu_t-\nu_t}_1 .
\]
Multiplying by \(e^{-\eta t}\), taking the essential supremum over \(t\ge0\),
and using Lemma~\ref{lem:hjb-infty} proves \eqref{eq:Q-infty-lip}.

The proof of \eqref{eq:policy-infty-lip} follows exactly as in the proof of
\eqref{eq:policy-lip} in Lemma~\ref{lem:q-policy}, by applying point-wise in
\((t,x)\) and then taking the essential supremum over \(t\ge0\).
\end{proof}

\begin{proof}[Proof of Theorem \ref{thm:infty-contraction}.]
Let \(\pmb \pi,\pmb \pi'\in\Pi_\infty\), and set
\(
  \pmb \mu=\Gamma_M^\infty(\pmb \pi),
\)
and
\(
  \pmb \nu=\Gamma_M^\infty(\pmb \pi').
\)
Using Lemma~\ref{lem:q-policy-infty} and Lemma~\ref{lem:forward-infty}, we
obtain
\[
\begin{aligned}
d_{\eta,\Pi}^{\infty}
(\Phi_\delta^\infty(\pmb \pi),\Phi_\delta^\infty(\pmb \pi'))
&\le
\alpha^{-1}
d_{\eta,Q}^{\infty}
(\pmb Q^{\delta,\infty,\pmb \mu},\pmb Q^{\delta,\infty,\pmb \nu})                 \\
&\le
\alpha^{-1}
C_\mu^{\delta,\infty}
\left(
    1+\frac{B_\Lambda}{\delta-B_\Lambda-\eta}
\right)
d_{\eta,1}^{\infty}(\pmb \mu,\pmb \nu)                                  \\
&\le
\frac{B_\Lambda}{\eta-a_M}
\cdot
\frac1\alpha
\cdot
C_\mu^{\delta,\infty}
\left(
    1+\frac{B_\Lambda}{\delta-B_\Lambda-\eta}
\right)
d_{\eta,\Pi}^{\infty}(\pmb \pi,\pmb \pi')                                \\
&=
\kappa_{\delta,\eta,\alpha}^{\infty}
d_{\eta,\Pi}^{\infty}(\pmb \pi,\pmb \pi').
\end{aligned}
\]
Thus \(\Phi_\delta^\infty\) is a contraction on the complete metric space \((\Pi_\infty,d^{\infty}_{\eta,\Pi})\) whenever
\(\kappa_{\delta,\eta,\alpha}^{\infty}<1\). Banach's fixed-point theorem therefore yields a unique fixed point and convergence of the Picard iteration of \(\pmb\pi^{k+1}=\Phi_\delta^\infty(\pmb\pi^k)\) from every initial policy. Together with \(\pmb\mu^{\delta,\infty}
=\Gamma_M^\infty(\pmb\pi^{\delta,\infty})\), this fixed point determines the unique infinite-horizon non-stationary \((\delta,\alpha)\)-regularized MFE.
\end{proof}

\begin{remark}
The infinite-horizon contraction condition has the same forward--backward structure as the finite-horizon condition and the only difference is that \(S_V^{\delta,\infty}\) contains no terminal-cost term.
Consequently, the infinite-horizon contraction condition is no stronger than
the corresponding finite-horizon condition obtained with the same Bielecki
weight.
\end{remark}

\subsection{Finite-to-infinite horizon error bounds}\label{subsec:finite-infty-error}
We now compare the finite-horizon equilibrium with the infinite-horizon equilibrium truncated to \([0,T]\). The argument first establishes a
truncation estimate for a fixed mean-field flow and then lifts it to the
equilibrium level using contraction stability.
The Bielecki metric is particularly suited to this comparison because its
exponential weight can be chosen independently of \(T\). In contrast, the
spectral metric used in Lemma~\ref{lem:contraction} depends on a principal
eigenfunction of the majorizing operator, while \eqref{eq:eig} characterizes
this eigenfunction only implicitly and provides no quantitative control.
The explicit structure of the Bielecki metric therefore allows us to derive
error bounds that are uniform across horizon lengths.
Since
\(
  -\log|\A|\le \Omega(q)\le 0,
\)
the regularized running cost \eqref{eq:infinite-discounted-cost} is uniformly bounded in absolute value by
\(
  M_{\alpha,c}:=M_c+\alpha\log|\A|.
\)
Define
\begin{equation}\label{eq:R-delta-alpha}
  R_{\delta,\alpha}
  :=
  M_g + M_{\alpha,c}\delta^{-1}.
\end{equation}

The next lemma gives the truncation error for a fixed mean-field flow.

\begin{lemma}\label{lem:finite-infty-fixed-flow}
Suppose that Assumption~\ref{ass:basic-reg} holds, and
\(\delta>B_\Lambda\). Let
\(\pmb \mu\in\mathcal M_\infty(\mu_0)\), and let
\(\pmb V^{\delta,T,\pmb \mu}\) and \(\pmb V^{\delta,\infty,\pmb \mu}\) be the finite-horizon and
infinite-horizon discounted regularized value functions generated by the same
mean-field flow \(\pmb\mu\).  Then, for a.e. \(t\in[0,T]\),
\begin{equation}\label{eq:fixed-flow-V-tail}
  \norm{
  V_t^{\delta,T,\pmb \mu}
  -
  V_t^{\delta,\infty,\pmb \mu}
  }_\infty
  \le
  R_{\delta,\alpha}
  e^{-(\delta-B_\Lambda)(T-t)} .
\end{equation}
Consequently,
\begin{equation}\label{eq:fixed-flow-Q-tail}
  \norm{
  Q_t^{\delta,T,\pmb \mu}
  -
  Q_t^{\delta,\infty,\pmb \mu}
  }_\infty
  \le
  B_\Lambda R_{\delta,\alpha}
  e^{-(\delta-B_\Lambda)(T-t)} .
\end{equation}
\end{lemma}

\begin{proof}
The boundedness of the regularized running cost and the discount factor imply:
\begin{equation}\label{eq:Vinfty-bound}
  \norm{V_t^{\delta,\infty,\pmb \mu}}_\infty
  \le
  M_{\alpha,c}\delta^{-1},
  \qquad t\ge0,
\end{equation}
which is a well-known result and is straightforward to prove.
Since \(V_T^{\delta,T,\pmb \mu}\equiv g\), it follows from \eqref{eq:Vinfty-bound} that
\begin{equation}\label{eq:terminal-tail-bound}
  \norm{
  V_T^{\delta,T,\pmb \mu}
  -
  V_T^{\delta,\infty,\pmb \mu}
  }_\infty
  \le
  M_g +M_{\alpha,c}\delta^{-1}
  =
  R_{\delta,\alpha}.
\end{equation}

We now compare the finite- and infinite- horizon HJB equations on \([0,T]\). Notice that for any \(T \in [1,\infty)\cup \{\infty\}\), it holds that
\[
  V_t^{\delta,T,\pmb \mu}(x)
  =
  e^{-\delta(T-t)}g(x)
  +
  \int_t^T
  e^{-\delta(s-t)}
  H_s^{\pmb \mu}(x,V_s^{\delta,T,\pmb \mu})
  \dd s.
\]
Subtracting the \(T < \infty\) and \(T = \infty\) cases from each other, it holds that
\[
\begin{aligned}
&
\left|
V_t^{\delta,T,\pmb \mu}(x)
-
V_t^{\delta,\infty,\pmb \mu}(x)
\right|                                                   \\
&\quad\le
e^{-\delta(T-t)}
\norm{
V_T^{\delta,T,\pmb \mu}
-
V_T^{\delta,\infty,\pmb \mu}
}_\infty
+
\int_t^T
e^{-\delta(s-t)}
\left|
H_s^{\pmb \mu}(x,V_s^{\delta,T,\pmb \mu})
-
H_s^{\pmb \mu}(x,V_s^{\delta,\infty,\pmb \mu})
\right|
\dd s .
\end{aligned}
\]
The Hamiltonian is \(B_\Lambda\)-Lipschitz in the value variable. 
Thus, by Gr\"onwall's inequality, one obtains
\begin{align}\label{eq:tail-gronwall-pre}
  \norm{
  V_t^{\delta,T,\pmb \mu}
  -
  V_t^{\delta,\infty,\pmb \mu}
  }_\infty
  &\le
  e^{-\delta(T-t)}\norm{
  V_T^{\delta,T,\pmb \mu}
  -
  V_T^{\delta,\infty,\pmb \mu}
  }_\infty
  +
  B_\Lambda
  \int_t^T e^{-\delta(s-t)}\norm{
  V_s^{\delta,T,\pmb \mu}
  -
  V_s^{\delta,\infty,\pmb \mu}
  }_\infty\,\dd s 
  \\&\le \norm{
  V_T^{\delta,T,\pmb \mu}
  -
  V_T^{\delta,\infty,\pmb \mu}
  }_\infty e^{-(\delta-B_\Lambda)(T-t)}.
\end{align} 
Using \eqref{eq:terminal-tail-bound}, we directly obtain \eqref{eq:fixed-flow-V-tail}.

Finally, for the same flow \(\mu\), the \(Q\)-functions differ only through
their value functions:
\begin{equation*}
\left|
Q_t^{\delta,T,\pmb \mu}(x,a)
-
Q_t^{\delta,\infty,\pmb \mu}(x,a)
\right| =
\left|
\sum_{y\in\X}
\Lambda(x,y,a,\mu_t)
\left(
V_t^{\delta,T,\pmb \mu}(y)
-
V_t^{\delta,\infty,\pmb \mu}(y)
\right)\right| \le
B_\Lambda
\norm{
V_t^{\delta,T,\pmb \mu}
-
V_t^{\delta,\infty,\pmb \mu}
}_\infty .
\end{equation*}
Combining this with \eqref{eq:fixed-flow-V-tail} proves
\eqref{eq:fixed-flow-Q-tail}.
\end{proof}

We now lift the fixed-flow truncation estimates presented in Lemma \ref{lem:finite-infty-fixed-flow} to the equilibrium level.  
The following theorem
quantifies the convergence of the finite-horizon MFE to the infinite-horizon
MFE on \([0,T]\), which is the main result of this section.

\begin{theorem}\label{thm:finite-infty-error}
Suppose that Assumption~\ref{ass:basic-reg} holds.  Assume that there exists \(\eta\) such that
\(
  a_M<\eta<\delta-B_\Lambda
\)
and
\(
\kappa_{\delta,\eta,\alpha}^T <1.
\)
Then, both the finite- and infinite-horizon \((\delta,\alpha)\)-regularized MFEs exist and are unique, and for every \(T<\infty\),
\begin{equation}\label{eq:policy-finite-infty-error}
  d_{\eta,\Pi}^{T}
  \bigl(
      \pmb \pi^{\delta,T},
      \pmb \pi^{\delta,\infty}|_{[0,T]}
  \bigr)
  \le
  \frac{B_\Lambda R_{\delta,\alpha}}
       {\alpha(1-\kappa_{\delta,\eta,\alpha}^T)}
  e^{-\eta T}.
\end{equation}
Moreover,
\begin{equation}\label{eq:mf-finite-infty-error}
  d_{\eta,1}^{T}
  \bigl(
      \pmb \mu^{\delta,T},
      \pmb \mu^{\delta,\infty}|_{[0,T]}
  \bigr)
  \le
  \frac{B_\Lambda}{\eta-a_M}
  \frac{B_\Lambda R_{\delta,\alpha}}
       {\alpha(1-\kappa_{\delta,\eta,\alpha}^T)}
  e^{-\eta T}.
\end{equation}
Equivalently, for a.e. \(0\le t\le T\),
\begin{align}
  \norm{
      \pi_t^{\delta,T}
      -
      \pi_t^{\delta,\infty}
  }_{\Pi}
  &\le
  \frac{B_\Lambda R_{\delta,\alpha}}
       {\alpha(1-\kappa_{\delta,\eta,\alpha}^T)}
  e^{-\eta(T-t)},\label{eq:policy-pointwise-finite-infty}\\
  \norm{
      \mu_t^{\delta,T}
      -
      \mu_t^{\delta,\infty}
  }_1
  &\le
  \frac{B_\Lambda}{\eta-a_M}
  \frac{B_\Lambda R_{\delta,\alpha}}
       {\alpha(1-\kappa_{\delta,\eta,\alpha}^T)}
  e^{-\eta(T-t)}.\label{eq:mf-pointwise-finite-infty}
\end{align}
\end{theorem}

\begin{proof}
The condition
\(\kappa_{\delta,\eta,\alpha}^T<1\) implies the finite-horizon contraction condition from Theorem~\ref{thrm:main}.  Since the infinite-horizon span
constant does not contain the terminal term \(\spn(g)\), the same condition also implies the infinite-horizon contraction condition in Theorem~\ref{thm:infty-contraction}.  Hence both equilibria are uniquely defined.

Set
\(
  \bar{\pmb \pi}^T:=\pmb \pi^{\delta,\infty}|_{[0,T]},
\)
and
\(
  \bar{\pmb \mu}^T:=\pmb \mu^{\delta,\infty}|_{[0,T]}.
\)
Since the forward equation is causal and both systems start from the same
\(\mu_0\), the finite-horizon forward operator applied to \(\bar{\pmb \pi}^T\)
produces exactly \(\bar{\pmb \mu}^T\).  Thus
\(
  \Gamma_M^T(\bar{\pmb \pi}^T)=\bar{\pmb \mu}^T.
\)

We first estimate how far \(\bar{\pmb \pi}^T\) is from being a fixed point of the
finite-horizon operator \(\Phi_\delta^T\).  
For a.e. \(t\in[0,T]\), since \(\pmb \pi^{\delta,\infty}\) is a fixed point of
\(\Phi_\delta^\infty\),
\(
  \bar\pi_t^T
  =
  \Gamma_\Pi^T
  \bigl(
      \pmb Q^{\delta,\infty,\pmb \mu^{\delta,\infty}}
  \bigr)_t.
\)
On the other hand, by the definition of the finite-horizon MFE operator
\(\Phi_\delta^T\),
\(
  \Phi_\delta^T(\bar{\pmb \pi}^T)_t
  =
  \Gamma_\Pi^T
  \bigl(
      \pmb Q^{\delta,T,\bar{\pmb \mu}^T}
  \bigr)_t.
\)
The two \(Q\)-functions above are generated by the same flow on \([0,T]\);
they differ only because one comes from the finite-horizon HJB equation and
the other from the infinite-horizon HJB equation.  By
Lemma~\ref{lem:finite-infty-fixed-flow},
\[
  \norm{
  Q_t^{\delta,T,\bar{\pmb \mu}^T}
  -
  Q_t^{\delta,\infty,\pmb \mu^{\delta,\infty}}
  }_\infty
  \le
  B_\Lambda R_{\delta,\alpha}
  e^{-(\delta-B_\Lambda)(T-t)} .
\]
Using the Lipschitz estimate for the soft-min map gives
\[
\begin{aligned}
  \norm{
  \Phi_\delta^T(\bar{\pmb \pi}^T)_t
  -
  \bar\pi_t^T
  }_{\Pi}
  &\le
  \alpha^{-1}
  \norm{
  Q_t^{\delta,T,\bar{\pmb \mu}^T}
  -
  Q_t^{\delta,\infty,\pmb \mu^{\delta,\infty}}
  }_\infty \le
  B_\Lambda R_{\delta,\alpha}\alpha^{-1}
  e^{-(\delta-B_\Lambda)(T-t)} .
\end{aligned}
\]
Multiplying by \(e^{-\eta t}\) yields
\[
  e^{-\eta t}
  \norm{
  \Phi_\delta^T(\bar{\pmb \pi}^T)_t
  -
  \bar\pi_t^T
  }_{\Pi}
  \le
  \frac{B_\Lambda R_{\delta,\alpha}}{\alpha}
  e^{-\eta t}
  e^{-(\delta-B_\Lambda)(T-t)} \le \frac{B_\Lambda R_{\delta,\alpha}}{\alpha}e^{-\eta T},
\]
where the last inequality follows from \(\eta<\delta-B_\Lambda\).
Taking the essential supremum over \(t\in[0,T]\), we obtain
\begin{equation}\label{eq:finite-residual}
  d_{\eta,\Pi}^{T}
  \bigl(
      \Phi_\delta^T(\bar{\pmb \pi}^T),
      \bar{\pmb \pi}^T
  \bigr)
  \le
  B_\Lambda R_{\delta,\alpha} \alpha^{-1}
  e^{-\eta T}.
\end{equation}
Now use fixed-point stability.  Since
\(\pmb \pi^{\delta,T}\) is a fixed point of \(\Phi_\delta^T\), we have
\(
  \pmb \pi^{\delta,T}=\Phi_\delta^T(\pmb \pi^{\delta,T}).
\)
Therefore,
\[
\begin{aligned}
  d_{\eta,\Pi}^{T}
  (
      \pmb \pi^{\delta,T},
      \bar{\pmb \pi}^T
  )
  &=
  d_{\eta,\Pi}^{T}
  (
      \Phi_\delta^T(\pmb \pi^{\delta,T}),
      \bar{\pmb \pi}^T
  )                                                     \\
  &\le
  d_{\eta,\Pi}^{T}
  (
      \Phi_\delta^T(\pmb \pi^{\delta,T}),
      \Phi_\delta^T(\bar{\pmb \pi}^T)
  )
  +
  d_{\eta,\Pi}^{T}
  (
      \Phi_\delta^T(\bar{\pmb \pi}^T),
      \bar{\pmb \pi}^T
  )                                                     \\
  &\le
  \kappa_{\delta,\eta,\alpha}^T
  d_{\eta,\Pi}^{T}
  (
      \pmb \pi^{\delta,T},
      \bar{\pmb \pi}^T
  )
  +
  B_\Lambda R_{\delta,\alpha}\alpha^{-1}
  e^{-\eta T},
\end{aligned}
\]
where the last line uses the contraction estimate for \(\Phi_\delta^T\) and
\eqref{eq:finite-residual}. Rearranging leads to \eqref{eq:policy-finite-infty-error}.

It remains to estimate the mean-field flows.  Applying the forward Lipschitz
estimate from Lemma~\ref{lem:forward} to
\(\pmb \pi^{\delta,T}\) and \(\bar{\pmb \pi}^T\), we get
\[
\begin{aligned}
  d_{\eta,1}^{T}
  \bigl(
      \pmb \mu^{\delta,T},
      \bar{\pmb \mu}^T
  \bigr)
  &=
  d_{\eta,1}^{T}
  \bigl(
      \Gamma_M^T(\pmb \pi^{\delta,T}),
      \Gamma_M^T(\bar{\pmb \pi}^T)
  \bigr) \le
  \frac{B_\Lambda}{\eta-a_M}
  d_{\eta,\Pi}^{T}
  \bigl(
      \pmb \pi^{\delta,T},
      \bar{\pmb \pi}^T
  \bigr).
\end{aligned}
\]
Combining this with \eqref{eq:policy-finite-infty-error} proves
\eqref{eq:mf-finite-infty-error}.

Finally, the point-wise estimates follow directly from the definitions of the
weighted metrics. Indeed,
\[
  e^{-\eta t}
  \norm{
      \pi_t^{\delta,T}
      -
      \pi_t^{\delta,\infty}
  }_{\Pi}
  \le
  d_{\eta,\Pi}^{T}
  \bigl(
      \pmb \pi^{\delta,T},
      \pmb \pi^{\delta,\infty}|_{[0,T]}
  \bigr),
\]which, together with \eqref{eq:policy-finite-infty-error}, gives
\eqref{eq:policy-pointwise-finite-infty}.
The proof of \eqref{eq:mf-pointwise-finite-infty} is identical.
\end{proof}

\begin{remark}
The error bound is strongest away from the terminal time.  For each fixed
\(t\), the difference between the finite-horizon and infinite-horizon MFE
decays exponentially as \(T\to\infty\).  Near \(t=T\), however, the
finite-horizon problem remains influenced by the imposed terminal cost \(g\). An unweighted uniform-in-time error bound over the entire interval \([0,T]\) is therefore not expected in general.
\end{remark}

\section{Discounted--Undiscounted Error Bounds for Finite-Horizon MFGs}\label{sect:undisc}

Throughout this section, we fix a horizon length \(T<\infty\). We study the perturbation introduced by discounting in finite-horizon regularized MFGs, where the undiscounted model is obtained by setting \(\delta=0\). We first bound the distance between the discounted and undiscounted best-response maps.
Under a contraction condition for the discounted MFE operator, this bound is then lifted to an MFE-level error estimate between the discounted MFE and any undiscounted MFE, if such an equilibrium exists.

For a fixed \(\pmb \mu\in\mathcal M_T(\mu_0)\), and for 
\((t,x)\in[0,T]\times\X\), we similarly define the undiscounted regularized cost by
\begin{equation}\label{eq:undisc-cost-functional}
  J^{0,T}_\alpha(t,x;\hat{\pmb \pi},\pmb \mu)
  :=
  \mathbb E^{\hat{\pmb \pi}}_{t,x}\bigg[
      \int_t^T
      \left(
      \sum_{a\in\A}\hat\pi_s(a\mid X_s)c(X_s,a,\mu_s)
      +\alpha\Omega(\hat\pi_s(\cdot\mid X_s))
      \right)\dd s
      +g(X_T)
  \bigg],
\end{equation}
whose corresponding undiscounted value function is
\begin{equation}\label{eq:undisc-value}
  V_t^{0,T,\pmb \mu}(x)
:=\inf_{\hat{\pmb \pi}\in\Pi_T}J^{0,T}_\alpha(t,x;\hat{\pmb \pi},\pmb \mu).
\end{equation}

We will use the undiscounted regularized finite-horizon operators, obtained by setting \(\delta=0\) in the corresponding finite-horizon regularized control problem in Subsection~\ref{subsec:infinite-ns}. They are denoted by
 \(\Gamma_\Pi^T, \Gamma_Q^{0,T}, \Gamma_M^T\) and are defined on the same spaces
\(\Pi_T,\mathcal M_T(\mu_0),\mathcal V_T,\mathcal Q_T\). We define
\(
  \Phi_0^T
  :=
  \Gamma_\Pi^T\circ\Gamma_Q^{0,T}\circ\Gamma_M^T
\). An 
 finite-horizon undiscounted regularized MFE is a pair \((\pmb \pi^{0,T},\pmb \mu^{0,T})\in\Pi_T\times\mathcal M_T(\mu_0)\) such that
\(
    \pmb \pi^{0,T}=\Phi_0^T(\pmb \pi^{0,T}),\)
and \(\pmb \mu^{0,T}=\Gamma_M^T(\pmb \pi^{0,T}).\)

The next proposition provides a uniform bound on the difference between the
discounted and undiscounted MFE operators that depends on the horizon length \(T\), which is the main result of this section.

\begin{proposition}\label{prop:map-gap}
Suppose that Assumption~\ref{ass:basic-reg} holds.
For every \(\eta\ge0\),
\begin{equation}\label{eq:map-gap-weighted}
  \sup_{\pmb \pi\in\Pi_T}
  d_{\eta,\Pi}^T(\Phi_\delta^T(\pmb \pi),\Phi_0^T(\pmb \pi))
  \le
  \frac{B_\Lambda}{\alpha}
  \delta
  \left(
      \frac{M_{\alpha,c}}{2}T^2+M_gT
  \right).
\end{equation}
\end{proposition}

\begin{proof}
Fix \(\pmb \pi\in\Pi_T\), and set
\(
    \pmb \mu:=\Gamma_M^T(\pmb \pi).
\)
We first compare the discounted and undiscounted control problems under this
fixed mean-field flow.  For any admissible deviating policy
\(\hat{\pmb \pi}\in\Pi_T\), define the regularized running cost along the controlled
trajectory by
\(
  \ell_s^{\hat{\pmb \pi},\pmb \mu}
  :=
  \sum_{a\in\A}\hat\pi_s(a\mid X_s)c(X_s,a,\mu_s)
  +
  \alpha\Omega(\hat\pi_s(\cdot\mid X_s)).
\)
By the definition of \(\Omega\), it holds that
\(
    |\ell_s^{\hat{\pmb \pi},\pmb \mu}|
    \le
    M_c+\alpha\log|\A|
    =
    M_{\alpha,c}.
\)
Moreover, \(|g(X_T)|\le M_g\) almost surely. For fixed \((\hat{\pmb \pi},\pmb \mu)\), the controlled
state dynamics in the discounted and undiscounted problems are the same; only
the weights in the objectives are different.  Therefore,
\begin{align*}
\left|
J^{\delta,T}_\alpha(t,x;\hat{\pmb \pi},\pmb \mu)
-
J^{0,T}_\alpha(t,x;\hat{\pmb \pi},\pmb \mu)
\right| & =
\left|
\mathbb E^{\hat{\pmb \pi}}_{t,x}
\left[
\int_t^T
\bigl(e^{-\delta(s-t)}-1\bigr)
\ell_s^{\hat{\pmb \pi},\pmb \mu}\dd s
+
\bigl(e^{-\delta(T-t)}-1\bigr)g(X_T)
\right]
\right|                                                        \\
&\quad \le
M_{\alpha,c}
\int_t^T
\bigl(1-e^{-\delta(s-t)}\bigr)\dd s
+
M_g\bigl(1-e^{-\delta(T-t)}\bigr).
\end{align*}
The right-hand side is independent of \(\hat{\pmb \pi}\).  Hence, it holds that
\begin{equation}\label{eq:value-discount-gap}
\norm{V_t^{\delta,T,\pmb \mu}-V_t^{0,T,\pmb \mu}}_\infty
\le
M_{\alpha,c}
\int_t^T
\bigl(1-e^{-\delta(s-t)}\bigr)\dd s
+
M_g\bigl(1-e^{-\delta(T-t)}\bigr).
\end{equation}

We now estimate the two terms on the right-hand side.  Let
\(\tau:=T-t\).  Then
\(
\int_t^T
\bigl(1-e^{-\delta(s-t)}\bigr)\dd s
=
\int_0^\tau
\bigl(1-e^{-\delta r}\bigr)\dd r .
\)
Using \(1-e^{-\delta r}\le \delta r\), we have the upper bound
\(
\int_t^T
\bigl(1-e^{-\delta(s-t)}\bigr)\dd s
\le
\int_0^\tau \delta r\,\dd r
=
\frac{\delta}{2}(T-t)^2.
\)
Similarly,
\(
    1-e^{-\delta(T-t)}
    \le
    \delta(T-t).
\)
Substituting these two estimates into \eqref{eq:value-discount-gap} leads to
\begin{equation*}
\norm{V_t^{\delta,T,\pmb \mu}-V_t^{0,T,\pmb \mu}}_\infty
\le
\delta
\left(
    \frac{M_{\alpha,c}}{2}(T-t)^2
    +
    M_g(T-t)
\right).
\end{equation*}

We next pass from values to \(Q\)-functions.  Since the discounted and
undiscounted \(Q\)-functions have the same running-cost term, they differ only
through their value functions:
\begin{align*}
|Q_t^{\delta,T,\pmb \mu}(x,a)-Q_t^{0,T,\pmb \mu}(x,a)|
&=
\left|
\sum_{y\in\X}
\Lambda(x,y,a,\mu_t)
\bigl(V_t^{\delta,T,\pmb \mu}(y)-V_t^{0,T,\pmb \mu}(y)\bigr)
\right| \le
B_\Lambda
\norm{V_t^{\delta,T,\pmb \mu}-V_t^{0,T,\pmb \mu}}_\infty.
\end{align*}
Therefore,
\(
\norm{Q_t^{\delta,T,\pmb \mu}-Q_t^{0,T,\pmb \mu}}_\infty
\le
B_\Lambda
\delta
\left(
    \frac{M_{\alpha,c}}{2}(T-t)^2
    +
    M_g(T-t)
\right).
\)
Finally, applying an analogue of soft-min Lipschitz estimate from
Lemma~\ref{lem:q-policy} yields
\begin{align*}
\norm{\Phi_\delta^T(\pmb \pi)_t-\Phi_0^T(\pmb \pi)_t}_{\Pi}
=
\norm{
\Gamma_\Pi^T(\pmb Q^{\delta,T,\pmb \mu})_t
-
\Gamma_\Pi^T(\pmb Q^{0,T,\pmb \mu})_t
}_{\Pi} &\le
\frac{1}{\alpha}
\norm{Q_t^{\delta,T,\pmb \mu}-Q_t^{0,T,\pmb \mu}}_\infty                     \\
&\le
\frac{B_\Lambda}{\alpha}
\delta
\left(
    \frac{M_{\alpha,c}}{2}(T-t)^2
    +
    M_g(T-t)
\right).
\end{align*}
Since \(\eta\ge0\), we have \(e^{-\eta t}\le1\).  Hence
\begin{align*}
d_{\eta,\Pi}^T(\Phi_\delta^T(\pmb \pi),\Phi_0^T(\pmb \pi))
=
\esssup_{0\le t\le T}
e^{-\eta t}
\norm{\Phi_\delta^T(\pmb \pi)_t-\Phi_0^T(\pmb \pi)_t}_{\Pi} \le
\frac{B_\Lambda}{\alpha}
\delta
\left(
    \frac{M_{\alpha,c}}{2}T^2+M_gT
\right).
\end{align*}
Taking the essential supremum over \(\pmb \pi\in\Pi_T\) proves \eqref{eq:map-gap-weighted}.
\end{proof}

The following corollary lifts the map perturbation bound to the MFE level by
using the contraction property of the discounted MFE operator.

\begin{corollary}\label{cor:mfe-gap}
Suppose that there exists an undiscounted regularized MFE \((\pmb \pi^{0,T},\pmb\mu^{0,T})\).
Then, under the conditions of Theorem~\ref{thrm:main},
\begin{equation}\label{eq:mfe-policy-gap}
  d_{\eta,\Pi}^T(\pmb \pi^{\delta,T},\pmb \pi^{0,T})
  \le
  \frac{B_\Lambda}{\alpha(1-\kappa_{\delta,\eta,\alpha}^T)}
  \delta
  \left(
      \frac{M_{\alpha,c}}{2}T^2+M_gT
  \right).
\end{equation}
Moreover,
\begin{equation}\label{eq:mfe-measure-gap}
  d_{\eta,1}^T(\pmb \mu^{\delta,T},\pmb \mu^{0,T})
  \le
  \frac{B_\Lambda}{\eta-a_M}
  \frac{B_\Lambda}{\alpha(1-\kappa_{\delta,\eta,\alpha}^T)}
  \delta
  \left(
      \frac{M_{\alpha,c}}{2}T^2+M_gT
  \right).
\end{equation}
\end{corollary}

\begin{proof}
By definition, it holds that
\( \pmb \pi^{\delta,T}=\Phi_\delta^T(\pmb \pi^{\delta,T}),\) and \( \pmb \pi^{0,T}=\Phi_0^T(\pmb \pi^{0,T}).\)
Therefore, by the triangle inequality, we obtain
\begin{equation}\label{eq:suyo10}
d_{\eta,\Pi}^T(\pmb \pi^{\delta,T},\pmb \pi^{0,T}) \le
d_{\eta,\Pi}^T
\left(
\Phi_\delta^T(\pmb \pi^{\delta,T}),
\Phi_\delta^T(\pmb \pi^{0,T})
\right)
+
d_{\eta,\Pi}^T
\left(
\Phi_\delta^T(\pmb \pi^{0,T}),
\Phi_0^T(\pmb \pi^{0,T})
\right).
\end{equation}
By Theorem~\ref{thm:contraction}, the discounted operator \(\Phi_\delta^T\) is a
contraction under \(d_{\eta,\Pi}^T\) with contraction constant
\(\kappa_{\delta,\eta,\alpha}^T\).  Hence
\begin{equation}\label{eq:suyo11}
d_{\eta,\Pi}^T
\left(
\Phi_\delta^T(\pmb \pi^{\delta,T}),
\Phi_\delta^T(\pmb \pi^{0,T})
\right)
\le
\kappa_{\delta,\eta,\alpha}^T
d_{\eta,\Pi}^T(\pmb \pi^{\delta,T},\pmb \pi^{0,T}).
\end{equation}
Thus, \eqref{eq:suyo10} and \eqref{eq:suyo11} yield
\[
d_{\eta,\Pi}^T(\pmb \pi^{\delta,T},\pmb \pi^{0,T})
\le
\kappa_{\delta,\eta,\alpha}^T
d_{\eta,\Pi}^T(\pmb \pi^{\delta,T},\pmb \pi^{0,T})
+
\sup_{\pmb \pi\in\Pi_T}
d_{\eta,\Pi}^T(\Phi_\delta^T(\pmb \pi),\Phi_0^T(\pmb \pi)).
\]
Rearranging and applying Proposition~\ref{prop:map-gap} gives
\eqref{eq:mfe-policy-gap}.
Finally, since it holds that
\( \pmb \mu^{\delta,T}=\Gamma_M^T(\pmb \pi^{\delta,T}),\) and  \( \pmb \mu^{0,T}=\Gamma_M^T(\pmb \pi^{0,T}),\)
Lemma~\ref{lem:forward} implies that
\[
  d_{\eta,1}^T(\pmb \mu^{\delta,T},\pmb \mu^{0,T})
  \le
  B_\Lambda(\eta-a_M)^{-1}
  d_{\eta,\Pi}^T(\pmb \pi^{\delta,T},\pmb \pi^{0,T}).
\]
Combining this estimate with \eqref{eq:mfe-policy-gap} proves
\eqref{eq:mfe-measure-gap}.
\end{proof}
\begin{remark}
This corollary shows that the discounted MFE provides a controlled
approximation of a finite-horizon undiscounted MFE. 
Notably, the undiscounted operator \(\Phi_0^T\) need not be contractive; the argument only
requires existence of an undiscounted MFE and contraction of the discounted operator.
\end{remark}
\section{Conclusion}
In this work, we have studied the contraction properties of continuous-time non-stationary discounted regularized MFGs in which the state process evolves according to an explicit Markov kernel. In the presence of a discount rate, we used the theory of positive operators to develop a quantitative method for calculating contraction rates in the finite-horizon setting. Furthermore, we have demonstrated that the asymptotic contraction condition in the finite-horizon case coincides with that of the infinite-horizon case.

In contrast to the discrete-time case, the stationary continuous-time case presents additional difficulties. In particular, one must solve \(F^{\pmb \pi}(\pmb \mu)=0\) (where \(\pmb \pi\) and \(\pmb \mu\) are time-homogeneous) at each step and then establish Lipschitz properties of the mean-field terms that satisfy this equation in terms of \(\pmb \pi\), which introduces further technical challenges. Thus, the stationary case requires new insights, which we leave as a future work.
\appendix
\section{A Primer on Positive Operators}
In this section, we will provide the definitions and statements that are used throughout the text related to positive operators. Let \((B,\|\cdot\|_B)\) be a Banach space over the real or the complex field.
\begin{definition}\label{def:a1}
     We say that \(K\subset B\) is a \emph{cone} if it is a closed convex set such that \(\lambda K \subset K\) for all \(\lambda \ge 0\) and \(K\cap (-K) = \{0\}\). If, in addition, \(K-K\) is dense in \(B\), then we say that \(K\) is a \emph{total cone}.
\end{definition}

\begin{definition}\label{def:a2}
    Let \(T:B \to B\) be an operator. We say that \(T\) is \emph{compact} if it maps bounded subsets of \(B\) to relatively compact subsets of \(B\).
\end{definition}

\begin{definition}\label{def:a3}
    Let \(K\) be a total cone of \(B\). Suppose that \(T:B\to B\) is a compact operator (not necessarily linear). We say that \(T\) is a \emph{positive operator} on \(B\) if \(T(K)\subset K\).
\end{definition}

\begin{definition}\label{def:a4}
    Let \(T:B\to B\) be an operator. The spectral radius of \(T\), \(\rho(T)\), is defined as 
    \(
    \rho(T) := \lim_{n \to \infty} \| T^n\|^{1/n}_{B \to B},
    \)
    where \(\| \cdot \|_{B \to B}\) is the operator norm induced by \(\|\cdot \|_B\), which is also known as Gelfand's formula
\end{definition}

\begin{theorem}[Krein--Rutman]\label{thrm:a5}
Let \(T:B\to B\) be a positive (non-zero) compact operator. If \(\rho(T)>0\), then the spectral radius of \(T\) on \(B\) is an eigenvalue of \(T\). Furthermore, there exists an eigenvector \(u \in K-\{0\}\) such that \(Tu=\rho(T)u\).
\end{theorem}
\begin{proof}
    See \cite[Theorem 19.3]{deimling2013nonlinear}.
\end{proof}
\printbibliography
\end{document}